\documentclass[11pt,draftcls,onecolumn]{IEEEtran}
\usepackage{graphics}
\usepackage{graphicx,times,epsf,bm,cite, amssymb, amsmath}
\usepackage{subfigure}
\usepackage{algorithm}
\usepackage{algorithmic}
\usepackage{fancyhdr}

\newcommand{\beq}{\begin{equation}}
\newcommand{\eeq}{\end{equation}}
\newcommand{\beqa}{\begin{eqnarray}}
\newcommand{\eeqa}{\end{eqnarray}}
\newcommand{\ba}{\begin{array}}
\newcommand{\ea}{\end{array}}
\newtheorem{proposition}{Proposition}
\newtheorem{theorem}{Theorem}
\newtheorem{remark}{Remark}
\newtheorem{lemma}{Lemma}
\newtheorem{definition}{Definition}
\newtheorem{upperbound}{Upper Bound}
\newtheorem{outerregion}{Outer Region}

\begin{document}

\title{Distributed Compression for the Uplink of a Backhaul-Constrained Coordinated Cellular Network}
\author{Aitor del Coso$^{*}$ and S\'{e}bastien Simoens\thanks{Aitor del Coso is with the Centre Tecnol\`{o}gic de Telecomunicacions de Catalunya (CTTC), Av. Canal Olímpic S/N, 08860, Castelldefels, Spain. E-mail: aitor.delcoso@cttc.es}\thanks{S\'{e}bastien Simoens is with
Motorola Labs PARIS, Parc Les Algorithmes, 91193, Saint-Aubin,
France. E-mail: simoens@motorola.com. }\thanks{This work was
partially supported by the internship program of Motorola Labs,
Paris. Also, it was partially funded by the European Comission
under projects COOPCOM (IST-033533) and NEWCOM++ (IST-216715).} }
\maketitle


\begin{abstract}
We consider a backhaul-constrained coordinated cellular network.
That is, a single-frequency network with $N+1$ multi-antenna base
stations (BSs) that cooperate in order to decode the users' data,
and that are linked by means of a common lossless backhaul, of
limited capacity $\mathrm{R}$. To implement receive cooperation,
we propose distributed compression: $N$ BSs, upon receiving their
signals, compress them using a multi-source lossy compression
code. Then, they send the compressed vectors to a central BS,
which performs users' decoding. Distributed Wyner-Ziv coding is
proposed to be used, and is optimally designed in this work. The
first part of the paper is devoted to a network with a unique
multi-antenna user, that transmits a predefined Gaussian
space-time codeword. For such a scenario, the compression
codebooks at the BSs are optimized, considering the user's
achievable rate as the performance metric. In particular, for $N =
1$ the optimum codebook distribution is derived in closed form,
while for $N>1$ an iterative algorithm is devised. The second part
of the contribution focusses on the multi-user scenario. For it,
the achievable rate region is obtained by means of the optimum
compression codebooks for sum-rate and weighted sum-rate,
respectively. 
\end{abstract}

\vspace{5mm} \textbf{EDICS:} WIN-INFO, MSP-CAPC, MSP-MULT,
WIN-CONT.

\newpage
\section{Introduction}

Inter-cell interference is one of the most limiting factors of
current cellular networks. It can be partially, but not totally,
mitigated resorting to frequency-division multiplexing, sectorized
antennas and fractional frequency reuse \cite{LEE86}. However, a
more spectrally efficient solution has been recently proposed:
coordinated cellular networks \cite{FOSCHINI06}. They consist of
single-frequency networks with base stations (BSs) cooperating in
order to transmit to and receive from the mobile terminals.
Beamforming mechanisms are thus deployed in the downlink, as well
as coherent detection in the uplink, to drastically augment the
system capacity\cite{KARAKAYALI06,SIMEONE07}. Hereafter, we only
focus on the uplink channel.

Preliminary studies on the uplink performance of coordinated
networks consider all BSs connected via a lossless backhaul with
unlimited capacity \cite{SIMEONE_07}\cite{KAMOUN07}. Accordingly,
the capacity region of the network equals that of a MIMO
multi-access channel, with a supra-receiver containing all the
antennas of all cooperative BSs \cite{TELATAR05}. Such an
assumption seems optimistic in short-mid term, as operators are
currently worried about the costs of upgrading their backhaul to
support \textit{e.g.}, HSPA traffic load. To deal with a realistic
backhaul constraint, two approaches have been proposed:
\textit{i}) \textit{distributed decoding}\cite{AKTAS07,GRANT04},
consisting on a demodulating scheme distributely carried out among
BSs, based on local decisions and belief propagation. Decoding
delay appears to be its main problem. \textit{ii})
\textit{Quantization} \cite{MARSCH07}, where BSs quantize their
observations and forward them to decoding unit. Its main
limitation relies on its inability to take profit of signal
correlation between antennas/BSs; thus, introduces redundancy into
the backhaul.


This paper considers a new approach for the network:
\textit{distributed compression}. The cooperative BSs, upon
receiving their signals, distributely compress them using a
multi-source lossy compression code \cite{OOHAMA97}. Then, via the
lossless backhaul, they transmit the compressed signals to the
central unit (also a BS); which decompresses them using its own
received signal as side information, and finally uses them to
estimate the users' messages. Distributed compression has been
already proposed for coordinated networks in
\cite{SANDEROVICH05,SANDEROVICH07,SANDEROVICH08}. However, in
those works, authors consider single-antenna BSs with ergodic
fading. We extend the analysis here to the multiple-antenna case
with time-invariant
fading. 



The compression of signals with side information at the decoder is
introduced by Wyner and Ziv in \cite{WYNER78_I,WYNER78}. They show
that side information at the encoder is useless (\textit{i.e.,}
the rate-distortion tradeoff remains unchanged) to compress a
single, Gaussian, source when it is available at the decoder
\cite[Section 3]{WYNER78}. Unfortunately, when considering
multiple (correlated) signals, independently compressed at
different BSs, and to be recovered at a central unit with side
information, such a statement can not be claimed. Indeed, this is
an open problem, for which it is not even clear when
source-channel separation applies \cite{GASTPAR06_liaison}. To the
best of authors knowledge, the scheme that performs best (in a
rate-distortion sense) for this problem is Distributed Wyner-Ziv
(D-WZ) compression \cite{GASTPAR04}. Such a compression is the
direct extension of Berger-Tung coding to the decoding side
information case \cite{TUNG_THESIS,BERGER07}. In turn, Berger-Tung
compression can be thought as the lossy counterpart of the
Slepian-Wolf lossless coding \cite{SLEPIAN73}. D-WZ coding is thus
the compresssion scheme proposed to be used, and is detailed in
the sequel.

\textbf{Summary of Contributions.} This paper considers a
single-frequency network with $N+1$ multi-antenna BSs. The first
base station, denoted $\mathrm{BS}_0$, is the central unit and
centralizes the users' decoding. The rest,
$\mathrm{BS}_1,\cdots,\mathrm{BS}_N$, are cooperative BSs, which
distributely compress their received signals using a D-WZ code,
and independently transmit them to $\mathrm{BS}_0$ via the common
backhaul of aggregate capacity $\mathrm{R}$. In the network,
\textit{time-invariant}, \textit{frequency-flat} channels are
assumed, as well as transmit and receive channel state information
(CSI) at the users and BSs, respectively.

The first part of the paper is devoted to a network with a single
user, equipped with multiple antennas. It aims at deriving the
optimum compression codebooks at the BSs, for which the user's
transmission rate is maximized. Our contributions are the
following:

\begin{itemize}
\item First, Sec. \ref{sec:compression of vectors} revisits
Wyner-Ziv coding \cite[Section 3]{WYNER78} and Distributed
Wyner-Ziv coding \cite{TUNG_THESIS}, and adapts them to our
compression scenario. 

\item For the single user transmitting a given Gaussian codeword,
Sec. \ref{sec:system model} proves that the optimum compression
codebooks at the BSs are Gaussian distributed. Accordingly, the
compression step is modelled by means of Gaussian "compression"
noise, added by the BSs on their observations before
retransmitting them to the central unit.

\item Considering a unique cooperative BS in the network
(\textit{i.e.,} $N=1$), Sec. \ref{sec:single basestation} derives
in closed form the optimum "compression" noise for which the
user's rate is maximized. We also show that conditional
Karhunen-Lo\`{e}ve transform plus independent Wyner-Ziv coding of
scalar streams is optimal.

\item The compression design is extended in Sec. \ref{sec:multiple
basestations} to arbitrary $N$ BSs. The optimum "compression"
noises (\textit{i.e.,} the optimum codebook distributions) are
obtained by means of an iterative algorithm, constructed using
dual decomposition theory and a non-linear block coordinate
approach \cite{BOYD,BERTSEKAS}. Due to the non-convexity of the
noises optimization, only local convergence is proven.
\end{itemize}

The second part of the paper extends the analysis to a network
where multiple users transmit simultaneously. For it, the
achievable rate region is described resorting to the weighted
sum-rate optimization:
\begin{itemize}
\item First, the sum-rate of the network is derived in Sec.
\ref{sec:multiuser}, adapting previous results a single-user.
Later, the weighted sum-rate, and its associated optimum
compression "noises", are obtained by means of an iterative
algorithm, constructed using dual decomposition and Gradient
Projection \cite{BERTSEKAS}.
\end{itemize}

\textbf{Notation.} $\bm{E}\left\{\cdot\right\}$ denotes
expectation. $\bm{A}^T$, $\bm{A}^\dagger$ and $a^*$ stand for the
transpose of $\bm{A}$, conjugate transpose of $\bm{A}$ and complex
conjugate of $a$, respectively. $[a]^+=\max\left\{a,0\right\}$.
$I\left(\cdot;\cdot\right)$ denotes mutual information,
$H\left(\cdot\right)$ entropy. The derivative of a scalar function
$f\left(\cdot\right)$ with respect to a complex matrix $\bm{X}$ is
defined as in \cite{MATRIXCOOKBOOK}, \textit{i.e.,}
$\left[\frac{\partial f}{\partial
\bm{X}}\right]_{i,j}=\frac{\partial f}{\partial
\left[\bm{X}\right]_{i,j}}$. In such a way, \textit{e.g.,}
$\frac{\partial \mathrm{tr}\left\{\bm{A}\bm{X}\right\}}{\partial
\bm{X}}=\bm{A}^T$. Moreover, we compactly write
$\bm{Y}_{1:N}=\left\{\bm{Y}_{1},\cdots,\bm{Y}_{N}\right\}$,
 $\bm{Y}_{\mathcal{G}}=\left\{\bm{Y}_{i}|i\in\mathcal{G}\right\}$
 and $\bm{Y}_{n}^c=\left\{\bm{Y}_{i}|i\neq n\right\}$. A sequence
 of vectors $\left\{\bm{Y}_i^t\right\}_{t=1}^n$ is compactly denoted by
 $\bm{Y}_i^n$.
Furthermore, to define block-diagonal matrices, we state
$\mathrm{diag}\left(\bm{A}_1,\cdots,\bm{A}_n\right)$, with
$\bm{A}_i$ square matrices. $\mathrm{coh}\left(\cdot\right)$
stands for convex hull. Finally, the covariance of random vector
$\bm{X}$ conditioned on random vector $\bm{Y}$ is denoted by
$\bm{R}_{\bm{X}|\bm{Y}}$ and computed $\bm{R}_{\bm{X}|\bm{Y}} =
\bm{E}\left\{\left(\bm{X}-\bm{E}\left\{\bm{X}|\bm{Y}\right\}\right)\left(\bm{X}-\bm{E}\left\{\bm{X}|\bm{Y}\right\}\right)^\dagger|\bm{Y}\right\}.$

\section{Compression of Vector Sources}\label{sec:compression of
vectors}
%

The aim of compression within coordinated networks is to make the
decoder extract the more mutual information from the reconstructed
signals. Known rate-distortion results apply to this goal as
follows.

\subsection{Single-Source Compression with Decoder Side Information}

Consider Fig. \ref{fig:fig2} with $N=1$. Let $\bm{Y}_1^n$ be a
zero-mean, temporally memoryless, Gaussian vector to be compressed
at $\mathrm{BS}_1$. Assume that it is the observation of the
signal transmitted by user $s$, \textit{i.e.}, $\bm{X}_s^n$.
$\mathrm{BS}_1$ compresses the signal and sends it to
$\mathrm{BS}_0$, which makes use of its side information
$\bm{Y}_0^n$ to decompress it. Finally, once reconstructed the
signal into vector $\hat{\bm{Y}}_1^n$, the decoder uses it to
estimate the message transmitted by the user. Wyner's results
\cite{WYNER78} apply to this problem as follows.


\begin{definition}[Single-source Compression Code]\label{Def_SS} A $\left(n,2^{n\rho}\right)$
compression code with side information at the decoder $\bm{Y}_0$ is
defined by two mappings, $f_n(\cdot)$ and $g_n(\cdot)$ and three
spaces $\mathcal{Y}_1,\hat{\mathcal{Y}}_1$ and $\mathcal{Y}_0$,
where
\begin{eqnarray}
&&f_n:\mathcal{Y}_1^n\rightarrow \left\{1,\cdots,2^{n\rho}\right\}\nonumber\\
&&g_n:\left\{1,\cdots,2^{n\rho}\right\}\times\mathcal{Y}_0^n\rightarrow\hat{\mathcal{Y}}_1^n\nonumber.
\end{eqnarray}
\end{definition}

\begin{proposition}[Wyner-Ziv Coding \cite{WYNER78}] \label{Def_wyner_single_source}Let the random vector $\hat{\bm{Y}}_1$ with conditional
probability $p\left(\hat{\bm{Y}}_1|\bm{Y}_1\right)$ satisfy the
Markov chain
$\bm{Y}_0\rightarrow\bm{Y}_1\rightarrow\hat{\bm{Y}}_1$, and let
$\bm{Y}_0$ and $\bm{Y}_1$ be jointly Gaussian. Then, considering a
sequence of compression codes $\left(n,2^{n\rho}\right)$ with side
information $\bm{Y}_0$ at the decoder:
\begin{eqnarray}\label{eq:wyner_single_source}
\frac{1}{n}I\left(\bm{X}_s^n;
\bm{Y}_0^n,g_n\left(\bm{Y}_0^n,f_n\left(\bm{Y}_1^n\right)\right)\right)=I\left(\bm{X}_s;
\bm{Y}_0,\hat{\bm{Y}}_1\right)
\end{eqnarray} as $n\rightarrow \infty$ if:
\begin{itemize}
\item the compression rate $\rho$ satisfies
\begin{eqnarray}\label{eq:wyner_single_source_rate_constraint}
I\left(\bm{Y}_1;\hat{\bm{Y}}_1|\bm{Y}_0\right)\leq\rho,
\end{eqnarray} \item the compression codebook $\mathfrak{C}$ consists of $2^{n\rho}$
random sequences $\hat{\bm{Y}}_1^n$ drawn \textit{i.i.d.} from
$\prod_{t=0}^np\left(\hat{\bm{Y}}_1\right)$, where
$p\left(\hat{\bm{Y}}_1\right)=\sum_{\bm{Y}_1}p\left({\bm{Y}}_1\right)p\left(\hat{\bm{Y}}_1|\bm{Y}_1\right)$,
\item the encoding $f_n\left(\cdot\right)$ outputs the bin-index
of codewords $\hat{\bm{Y}}_1^n$ that are jointly typical with the
source sequence $\bm{Y}_1^n$. In turn, $g_n\left(\cdot\right)$
outputs the codeword $\hat{\bm{Y}}_1^n$ that, belonging to the bin
selected by the encoder, is jointly typical with ${\bm{Y}}_0^n$.
\end{itemize}
\end{proposition}

\begin{proof}
The proposition is proven in \cite[Lemma 5]{WYNER78} using joint
typicality arguments. 
\end{proof}


\subsection{Multiple-Source Compression with Decoder Side Information}

Consider Fig. \ref{fig:fig2}. Let $\bm{Y}_i^n, \ i=1,\cdots,N$ be
$N$ zero-mean, temporally memoryless, Gaussian vectors to be
compressed independently at $\mathrm{BS}_1,\cdots,\mathrm{BS}_N$,
respectively. Assume that they are the observations at the BSs of
the signal transmitted by user $s$, \textit{i.e.,} $\bm{X}_s^n$.
The compressed vectors are sent to $\mathrm{BS}_0$, which
decompresses them using its side information $\bm{Y}_0^n$ and uses
them to estimate the user's message. Notice that the architecture
in Fig. \ref{fig:fig2} imposes source-channel separation at the
compression step, which is not shown to be optimal. However, it
includes the coding scheme with best known performance:
Distributed Wyner-Ziv coding \cite{GASTPAR04}. It applies to the
setup as follows.

\begin{definition}[Multiple-source Compression Code]\label{Def_MS} A $\left(n,2^{n\rho_1},\cdots,2^{n\rho_N}\right)$
compression code with side information at the decoder $\bm{Y}_0$
is defined by $N+1$ mappings, $f_n^i(\cdot)$, $i=1,\cdots,N$, and
$g_n(\cdot)$, and $2N+1$ spaces
$\mathcal{Y}_i,\hat{\mathcal{Y}}_i$, $i=1,\cdots,N$ and
$\mathcal{Y}_0$, where
\begin{eqnarray}
&&f_n^i:\mathcal{Y}_i^n\rightarrow\left\{1,\cdots,2^{n\rho_i}\right\},\ \ i=1,\cdots,N \nonumber\\
&&g_n:\left\{1,\cdots,2^{n\rho_1}\right\}\times\cdots\times
\left\{1,\cdots,2^{n\rho_N}\right\}\times\mathcal{Y}_0^n\rightarrow\hat{\mathcal{Y}}_1^n\times\cdots\times\hat{\mathcal{Y}}_N^n\nonumber.
\end{eqnarray}
\end{definition}

\begin{proposition}[Distributed Wyner-Ziv Coding \cite{GASTPAR04}] \label{Def_berger_tung} Let the random vectors
$\hat{\bm{Y}}_i$, $i=1,\cdots,N$, have conditional probability
$p\left(\hat{\bm{Y}}_i|\bm{Y}_i\right)$ and satisfy the Markov
chain $\left(\bm{Y}_0,\bm{Y}_i^c,
\hat{\bm{Y}}_i^c\right)\rightarrow\bm{Y}_i\rightarrow\hat{\bm{Y}}_i$.
Let $\bm{Y}_0$ and $\bm{Y}_i$, $i=1,\cdots,N$ be jointly Gaussian.
Then, considering a sequence of compression codes
$\left(n,2^{n\rho_1},\cdots,2^{n\rho_N}\right)$ with side
information $\bm{Y}_0$ at the decoder:
\begin{eqnarray}\label{eq:wyner_multiple_source}
\frac{1}{n}I\left(\bm{X}_s^n;
\bm{Y}_0^n,g_n\left(\bm{Y}_0^n,f_n^1\left(\bm{Y}_1^n\right),\cdots,f_n^N\left(\bm{Y}_N^n\right)\right)\right)=I\left(\bm{X}_s;
\bm{Y}_0,\hat{\bm{Y}}_{1:N}\right)
\end{eqnarray} as $n\rightarrow \infty$ if:
\begin{itemize} \item the
compression rates $\rho_1,\cdots,\rho_N$ satisfy
\begin{eqnarray}\label{eq:wyner_multiple_source_rate_constraint}
I\left(\bm{Y}_\mathcal{G};\hat{\bm{Y}}_\mathcal{G}|\bm{Y}_0,\hat{\bm{Y}}_\mathcal{G}^c\right)\leq\sum_{i\in\mathcal{G}}\rho_i
\qquad \forall \mathcal{G}\subseteq\left\{1,\cdots,N\right\},
\end{eqnarray} \item each compression codebook $\mathfrak{C}_i$, $i=1,\cdots,N$ consists of $2^{n\rho_i}$
random sequences $\hat{\bm{Y}}_i^n$ drawn \textit{i.i.d.} from
$\prod_{t=1}^np\left(\hat{\bm{Y}}_i\right)$, where
$p\left(\hat{\bm{Y}}_i\right)=\sum_{\bm{Y}_i}p\left({\bm{Y}}_i\right)p\left(\hat{\bm{Y}}_i|\bm{Y}_i\right)$.
\item for every $i=1,\cdots,N$, the encoding
$f_n^i\left(\cdot\right)$ outputs the bin-index of codewords
$\hat{\bm{Y}}_i^n$ that are jointly typical with the source
sequence $\bm{Y}_i^n$. In turn, $g_n\left(\cdot\right)$ outputs
the codewords $\hat{\bm{Y}}_i^n$, $i=1,\cdots,N$ that, belonging
to the bins selected by the encoders, are all jointly typical with
${\bm{Y}}_0^n$.
\end{itemize}
\end{proposition}

\begin{proof}The proposition is proven for discrete sources and discrete
side information in \cite[Theorem 2]{GASTPAR04}. Also, the
extension to the Gaussian case is conjectured therein. The
conjecture can be proven by noting that D-WZ coding is equivalent
to Berger-Tung coding with side information at the
decoder\cite{TUNG_THESIS}. In turn, Berger-Tung coding can be
implemented through time-sharing of successive Wyner-Ziv
compressions \cite{BERGER07}, for which introducing side
information $\bm{Y}_0$ at the decoder reduces the compression rate
as in (\ref{eq:wyner_multiple_source_rate_constraint}). Due to
space limitations, we limit the proof to this sketch.
\end{proof}

Now, we can present the coordinated cellular network with D-WZ
coding.

\section{System Model}\label{sec:system model}

Let a single source $s$, equipped with $N_t$ antennas, transmit
data to base stations $\mathrm{BS}_0,\cdots,\mathrm{BS}_N$, each
one equipped with $N_i, \ i=1,\cdots,N$ antennas. The BSs, as in
typical 3G networks, are connected (through radio network
controllers) to a common lossless backhaul of aggregate capacity
$\mathrm{R}$, and $\mathrm{BS}_0$ is selected to be the decoding
unit. This user-to-BSs assignment is assumed to be given by upper
layers and out of the scope of the paper\footnote{The derivation
of the optimum set of BSs to decode the user is out of the scope
of our study. We refer the reader to \textit{e.g,} \cite{KAMOUN07}
for assignment algorithms and selection criteria.}.



The source transmits a message
$\omega\in\left\{1,\cdots,2^{nR_s}\right\}$ mapped onto a
zero-mean, Gaussian codeword $\bm{X}_s^n$, drawn \textit{i.i.d.}
from random vector
$\bm{X}_s\sim\mathcal{CN}\left(\bm{0},\bm{Q}\right)$ and not
subject to optimization. The transmitted signal, affected by
\textit{time-invariant}, \textit{memory-less} fading, is received
at the BSs under additive noise:
\begin{eqnarray}\label{signal_model_single_user}
\bm{Y}_i^n= \bm{H}_{s,i}\cdot\bm{X}_s^n + \bm{Z}_i^n, \ \
i=0,\cdots,N
\end{eqnarray}where $\bm{H}_{s,i}$ is the MIMO channel matrix between
user $s$ and $\mathrm{BS}_i$, and
$\bm{Z}_i\sim\mathcal{CN}\left(0,\sigma_r^2\bm{I}\right)$ is AWGN.
Channel coefficients are known at both the BSs and at the user,
while $\mathrm{BS}_0$ has centralized knowledge of all the
channels within the network.

\subsection{Problem Statement}

Base stations $\mathrm{BS}_1,\cdots,\mathrm{BS}_N$, upon receiving
their signals, distributely compress them using a D-WZ compression
code. Later, they transmit the compressed vectors to
$\textrm{BS}_0$, which recovers them and uses them to decode.
Considering so, the user's message can be reliably decoded
\textit{iif} \cite[Theorem 1]{SANDEROVICH05}:
\begin{eqnarray}
R_s &\leq& \lim_{n\rightarrow \infty}\frac{1}{n}I\left(\bm{X}_s^n;
\bm{Y}_0^n,g_n\left(\bm{Y}_0^n,f_n^1\left(\bm{Y}_1^n\right),\cdots,f_n^N\left(\bm{Y}_N^n\right)\right)\right)\\
&=&I\left(\bm{X}_s; \bm{Y}_0, \hat{\bm{Y}}_{1:N}\right).\nonumber
\end{eqnarray}
Second equality follows from (\ref{eq:wyner_multiple_source}) in
Prop. \ref{Def_berger_tung}. However, equality only holds for
compression rates satisfying the set of constraints
(\ref{eq:wyner_multiple_source_rate_constraint}). As mentioned, in
the backhaul there is only an aggregate rate constraint
$\mathrm{R}$, \textit{i.e.,} $\sum_{i\in\mathcal{G}}\rho_i\leq
\mathrm{R}$, $\forall
\mathcal{G}\subseteq\left\{1,\cdots,N\right\}$. Therefore, the set
of constraints (\ref{eq:wyner_multiple_source_rate_constraint})
can be all re-stated as:
\begin{eqnarray}\label{eq:wyner_multiple_source_rate_constraint2}
I\left(\bm{Y}_\mathcal{G};\hat{\bm{Y}}_\mathcal{G}|\bm{Y}_0,\hat{\bm{Y}}_\mathcal{G}^c\right)\leq\mathrm{R}
\qquad \forall \mathcal{G}\subseteq\left\{1,\cdots,N\right\}.
\end{eqnarray}Furthermore, from the Markov chain in Prop.
\ref{Def_berger_tung}, the following inequality holds
\begin{eqnarray}
I\left(\bm{Y}_\mathcal{G};\hat{\bm{Y}}_\mathcal{G}|\bm{Y}_0,\hat{\bm{Y}}_\mathcal{G}^c\right)\leq
I\left(\bm{Y}_{1:N};\hat{\bm{Y}}_{1:N}|\bm{Y}_0\right) \qquad
\forall \mathcal{G}\subseteq\left\{1,\cdots,N\right\}.
\end{eqnarray}Therefore, forcing the constraint
$I\left(\bm{Y}_{1:N};\hat{\bm{Y}}_{1:N}|\bm{Y}_0\right)\leq
\mathrm{R}$ to hold makes all constraints in
(\ref{eq:wyner_multiple_source_rate_constraint2}) to hold too.
Accordingly, the maximum transmission rate $\mathcal{C}$ of user
$s$ is obtained from optimization:
\begin{eqnarray}\label{eq:rate_def}
\mathcal{C}&=&\max_{\prod_{i=1}^Np\left(\hat{\bm{Y}}_{i}|\bm{Y}_{i}\right)}\ I\left(\bm{X}_s;\bm{Y}_0,\hat{\bm{Y}}_{1:N}\right)\\
&& \ \ \  \ \mathrm{s.t.} \ \
I\left(\bm{Y}_{1:N};\hat{\bm{Y}}_{1:N}|\bm{Y}_0\right)\leq
\mathrm{R},\nonumber
\end{eqnarray}%
%

\begin{theorem}\label{prop1}
Let $\bm{X}_s\sim\mathcal{CN}\left(\bm{0},\bm{Q}\right)$.
Optimization (\ref{eq:rate_def}) is solved for Gaussian
conditional distributions
$p\left(\hat{\bm{Y}}_{i}|\bm{Y}_{i}\right),\ i=1,\cdots,N$. Thus,
the compressed vectors can be modelled as $\hat{\bm{Y}}_i =
\bm{Y}_i + \bm{Z}_i^c$, where
$\bm{Z}_i^c\sim\mathcal{CN}\left(0,\bm{\Phi}_i\right)$ is
independent, Gaussian, "compression" noise at $\textrm{BS}_i$.
That is,
\begin{eqnarray}\label{eq:rate_prop1}
\mathcal{C}=\max_{\bm{\Phi}_1,\cdots,\bm{\Phi}_N\succeq0}\
\log\det\left(\bm{I}+\frac{\bm{Q}}{\sigma_r^2}\bm{H}_{s,0}^\dagger\bm{H}_{s,0}+\bm{Q}\sum_{n=1}^N\bm{H}_{s,n}^\dagger\left(\sigma_r^2\bm{I}+\bm{\Phi}_n\right)^{-1}\bm{H}_{s,n}
\right)\\
 \mathrm{s.t.} \ \
\log\det\left(\bm{I}+\textrm{diag}\left(\bm{\Phi}_1^{-1},\cdots,\bm{\Phi}_N^{-1}\right)\bm{R}_{\bm{Y}_{1:N}|\bm{Y}_0}\right)\leq
\mathrm{R}.\qquad \qquad \qquad \ \ \nonumber
\end{eqnarray}where the conditional covariance $\bm{R}_{\bm{Y}_{1:N}|\bm{Y}_0}$ follows
(\ref{app:cond_cova_singleuser_Y1:NYo_final}).
\end{theorem}
\begin{proof} See Appendix \ref{appen:prop1} for the proof.\end{proof}
\begin{remark} The maximization above is not concave in standard form: although the feasible
set is convex, the objective function is not concave on
$\bm{\Phi}_1,\cdots,\bm{\Phi}_N$.
\end{remark}

\subsection{Useful Upper Bounds}

Prior to solving (\ref{eq:rate_prop1}), we present two upper
bounds on it.

\begin{upperbound}\label{upperbound1} The achievable rate $\mathcal{C}$ in
(\ref{eq:rate_prop1}) is upper bounded by
\begin{eqnarray}
\mathcal{C} \leq I\left(\bm{X}_s;\bm{Y}_0,{\bm{Y}}_{1:N}\right)
=
\log\det\left(\bm{I}+\frac{\bm{Q}}{\sigma_r^2}\sum_{n=0}^N\bm{H}_{s,n}^\dagger
\bm{H}_{s,n}\right).
\end{eqnarray}
\end{upperbound}

\begin{upperbound}\label{upperbound2} The achievable rate
$\mathcal{C}$ in (\ref{eq:rate_prop1}) satisfies
\begin{eqnarray}
\mathcal{C} \leq I\left(\bm{X}_s;\bm{Y}_0\right) +
\mathrm{R}
=
\log\det\left(\bm{I}+\frac{1}{\sigma_r^2}\bm{H}_{s,0}\bm{Q}\bm{H}_{s,0}^\dagger
\right) + \mathrm{R}.
\end{eqnarray}
\end{upperbound}

\begin{proof} See Appendix \ref{appen:ub2} for the proof. \end{proof}

\begin{remark} Notice that, independently of the
number of BSs, the achievable rate is bounded above by the
capacity with $\mathrm{BS}_0$ plus the backhaul rate.
\end{remark}

\section{The Two-Base Stations Case }\label{sec:single basestation}

We first solve (\ref{eq:rate_prop1}) for $N=1$. As mentioned, the
objective function, which has to be maximized, is convex on
$\bm{\Phi}_1\succeq 0$. In order to make it concave, we change the
variables $\bm{\Phi}_1 = \bm{A}_1^{-1}$, so that
\begin{eqnarray}\label{eq:rate_singleBS_2}
\mathcal{C}&=&\max_{\bm{A}_1\succeq0}\
\log\det\left(\bm{I}+\frac{\bm{Q}}{\sigma_r^2}\bm{H}_{s,0}^\dagger\bm{H}_{s,0}+\bm{Q}\bm{H}_{s,1}^\dagger\left(\bm{A}_1\sigma_r^2+\bm{I}\right)^{-1}\bm{A}_1\bm{H}_{s,1}
\right)\\
&& \  \mathrm{s.t.} \ \
\log\det\left(\bm{I}+\bm{A}_1\bm{R}_{\bm{Y}_1|\bm{Y}_0}\right)\leq
\mathrm{R}.\nonumber
\end{eqnarray} The objective has turned into concave. However, the constraint now
does not define a convex feasible set. Therefore,
Karush-Kuhn-Tucker (KKT) conditions become
necessary\footnote{Notice that all feasible points are regular.}
but not sufficient for optimality. To solve the problem, we need
to resort to the general sufficiency condition \cite[Proposition
3.3.4]{BERTSEKAS}: first, we derive a matrix $\bm{A}_1^*$ for
which the KKT conditions hold. Later, we demonstrate that the
selected matrix also satisfies the general sufficiency condition,
thus becoming the optimal solution. The optimum compression noise
is finally recovered as
$\bm{\Phi}_1^*=\left(\bm{A}_1^*\right)^{-1}$. This result is
presented in Theorem \ref{Prop2}:

\begin{theorem}\label{Prop2}
Let $\bm{X}_s\sim \mathcal{CN}\left(\bm{0},\bm{Q}\right)$ and the
conditional covariance (see Appendix \ref{cova_single_user}):
\begin{eqnarray}
\bm{R}_{\bm{Y}_1|\bm{Y}_0}
=\bm{H}_{s,1}\left(\bm{I}+\frac{\bm{Q}}{\sigma_r^2}\bm{H}_{s,0}^\dagger\bm{H}_{s,0}\right)^{-1}\bm{Q}\bm{H}_{s,1}^\dagger+\sigma_r^2\bm{I},
\end{eqnarray} with eigen-decomposition
$\bm{R}_{\bm{Y}_1|\bm{Y}_0}=\bm{U}\textrm{diag}\left(s_1,\cdots,s_{N_1}\right)\bm{U}^\dagger$.
The optimum "compression" noise at $\mathrm{BS}_1$ is
$\bm{\Phi}_1^*=\bm{U}\left(\textrm{diag}\left(\eta_1,\cdots,\eta_{N_1}\right)\right)^{-1}\bm{U}^\dagger
$, with
\begin{eqnarray}\label{eq:prop3}
\eta_j=\left[\frac{1}{\lambda}\left(\frac{1}{\sigma_r^2}-\frac{1}{s_j}\right)-\frac{1}{\sigma_r^2}\right]^+,
\end{eqnarray} and $\lambda$ is such that
$\sum_{j=1}^{N_1}\log\left(1+\eta_j s_j\right)=\mathrm{R}$.
\end{theorem}

\begin{proof}See Appendix \ref{append:prop2} for the proof\end{proof}


\subsection{Practical Implementation} The optimum
compression in Theorem \ref{Prop2} can be carried out using a
practical Transform Coding (TC) approach. With TC, $\mathrm{BS}_1$
first transforms its received vector using an invertible linear
function and then separately compresses the resulting scalar
streams \cite{GOYAL01}. We show that the conditional
Karhunen-Lo\`{e}ve transform (CKLT) is an optimal linear
transformation \cite{GASTPAR06}. First, let recall that
multiplying a vector by a matrix does not change the mutual
information \cite{COVER}, \textit{i.e.},
$I\left(\bm{X}_s;\bm{Y}_0,\hat{\bm{Y}_1}\right)=I\left(\bm{X}_s;\bm{Y}_0,\bm{U}^\dagger\hat{\bm{Y}_1}\right)$
and
$I\left(\bm{Y}_1;\hat{\bm{Y}_1}|\bm{Y}_0\right)=I\left(\bm{Y}_1;\bm{U}^\dagger\hat{\bm{Y}_1}|\bm{Y}_0\right)
$. From Theorem \ref{Prop2}, the optimum compressed vector
satisfies $\hat{\bm{Y}}_1^* = \bm{Y}_1 + \bm{Z}_c^{*}$,
 with
$\bm{Z}_c^*\sim\mathcal{CN}\left(0,\bm{U}\bm{\eta}^{-1}\bm{U}^\dagger\right)$
and $\bm{R}_{\bm{Y}_1|\bm{Y}_0}=\bm{U}\bm{S}\bm{U}^\dagger$.
Therefore, the following compressed vectors are also optimal
\begin{eqnarray}\label{single_base_station_y_hat}
\hat{\bm{Y}}_1 = \bm{U}^\dagger\bm{Y}_1 +
\bm{U}^\dagger\bm{Z}_c^{*},
\end{eqnarray}where vector $\bm{U}^\dagger\bm{Y}_1$ is referred to as the CKLT of vector $\bm{Y}_1$. Notice now that
$\bm{R}_{\hat{\bm{Y}}_1|\bm{Y}_0}=\bm{R}_{\bm{U}^\dagger\bm{Y}_1|\bm{Y}_0}+\bm{R}_{\bm{U}^\dagger\bm{Z}_c^*}=\bm{S}+\bm{\eta}^{-1}$
is diagonal. Therefore, the elements of the compressed vector
$\hat{\bm{Y}}_1$ are conditionally uncorrelated given $\bm{Y}_0$.
Likewise, so are the elements of vector $\bm{U}^\dagger\bm{Y}_1$.
Due to this uncorrelation, each element $j=1,\cdots,N_1$ of vector
$\bm{U}^\dagger{\bm{Y}}_1$ can be compressed, without loss of
optimality, independently of the compression of the others
elements, at a compression rate
$r_j=\log\left(1+\eta_js_j\right)$, $j=1,\cdots,N_1$
\cite{WYNER78}. From Theorem \ref{Prop2} we validate that
$\sum_{j=1}^{N_1}r_j=\mathrm{R}$. This demonstrates that CKLT plus
independent coding of streams is optimal, not only for minimizing
distortion as shown in \cite{GASTPAR06}, but also for maximizing
the achievable rate of coordinated networks.

\section{The Multiple-Base Stations Case}\label{sec:multiple basestations}

Consider now $\mathrm{BS}_0$ assisted by $N>1$ cooperative BSs.
The achievable rate follows (\ref{eq:rate_prop1}) where, as
previously, the objective function is not concave over
$\bm{\Phi}_n$ , $n=1,\cdots,N$. To make it concave, we change the
variables: $\bm{\Phi}_n=\bm{A}_n^{-1}$, $n=1,\cdots,N$, so that:
\begin{eqnarray}\label{eq:rate_multiple_base_stations}
\mathcal{C}=\max_{\bm{A}_1,\cdots,\bm{A}_N\succeq0}\
\log\det\left(\bm{I}+\frac{\bm{Q}}{\sigma_r^2}\bm{H}_{s,0}^\dagger\bm{H}_{s,0}+\bm{Q}\sum_{n=1}^N\bm{H}_{s,n}^\dagger\left(\bm{A}_n\sigma_r^2+\bm{I}\right)^{-1}\bm{A}_n\bm{H}_{s,n}
\right)\\
\  \ \mathrm{s.t.} \ \
\log\det\left(\bm{I}+\textrm{diag}\left(\bm{A}_1,\cdots,\bm{A}_N\right)\bm{R}_{\bm{Y}_{1:N}|\bm{Y}_0}\right)\leq
\mathrm{R}.\qquad \qquad \qquad \qquad \ \ \ \nonumber
\end{eqnarray}Again, the feasible set does not define a convex set. Our strategy to solve
the optimization is the following: first, we show that the duality
gap for the problem is zero. Later, we propose an iterative
algorithm that solves the dual problem, thus solving the primal
too. An interesting property of the dual problem is that the
coupling constraint in (\ref{eq:rate_multiple_base_stations}) is
decoupled \cite[Chapter 5]{BERTSEKAS}.

\subsection{The dual problem}
Let the Lagrangian of (\ref{eq:rate_multiple_base_stations}) be
defined on $\bm{A}_n\succeq0$, $n=1,\cdots,N$ and $\lambda\geq0$
as:
\begin{eqnarray}\label{eq:f_dual}
\mathcal{L}\left(\bm{A}_1,\cdots,\bm{A}_N,\lambda\right) &=&
\log\det\left(\bm{I}+\frac{\bm{Q}}{\sigma_r^2}\bm{H}_{s,0}^\dagger\bm{H}_{s,0}+\bm{Q}\sum_{n=1}^N\bm{H}_{s,n}^\dagger\left(\bm{A}_n\sigma_r^2+\bm{I}\right)^{-1}\bm{A}_n\bm{H}_{s,n}
\right) \nonumber\\
&& -
\lambda\cdot\left(\log\det\left(\bm{I}+\textrm{diag}\left(\bm{A}_1,\cdots,\bm{A}_N\right)\bm{R}_{\bm{Y}_{1:N}|\bm{Y}_0}\right)-\mathrm{R}\right).
\end{eqnarray} The dual function $g\left(\lambda\right)$
for $\lambda\geq 0$ follows \cite[Section 5.1]{BOYD}:
\begin{eqnarray}\label{eq:g_dual}
g\left(\lambda\right) = \max_{\bm{A}_1,\cdots,\bm{A}_N\succeq0}
\mathcal{L}\left(\bm{A}_1,\cdots,\bm{A}_N,\lambda\right) .
\end{eqnarray}
The solution of the dual problem is then obtained from
\begin{eqnarray}\label{eq:rate_dual}
\mathcal{C}'=\min_{\lambda\geq0}g\left(\lambda\right).
\end{eqnarray}

\begin{lemma} The duality gap for optimization (\ref{eq:rate_multiple_base_stations}) is zero,
\textit{i.e.,} the primal problem
(\ref{eq:rate_multiple_base_stations}) and the dual problem
(\ref{eq:rate_dual}) have the same solution.
\end{lemma}

\begin{proof}
The duality gap for problems of the form of
(\ref{eq:rate_multiple_base_stations}), and satisfying the
time-sharing property, is zero \cite[Theorem 1]{YU06}.
Time-sharing property is defined as follows: let $\mathcal{C}_x,
\mathcal{C}_y, \mathcal{C}_z$ be the solution of
(\ref{eq:rate_multiple_base_stations}) for backhaul rates
$\mathrm{R}_x,\mathrm{R}_y, \mathrm{R}_z$, respectively. Consider
$\mathrm{R}_z = \nu\mathrm{R}_x+ \left(1-\nu\right)\mathrm{R}_y$
for some $0\leq \nu\leq 1$. Then, the property is satisfied if and
only if $\mathcal{C}_z \geq \nu\mathcal{C}_x+
\left(1-\nu\right)\mathcal{C}_y$, $\forall \ \nu \in
\left[0,1\right]$. That is, if the solution of
(\ref{eq:rate_multiple_base_stations}) is concave with respect to
the backhaul rate $\mathrm{R}$. It is well known that time-sharing
of compressions cannot decrease the resulting distortion
\cite[Lemma 13.4.1]{COVER}, neither improve the mutual information
obtained from the reconstructed vectors. Hence, the property holds
for (\ref{eq:rate_multiple_base_stations}), and the duality gap is
zero.
\end{proof}

We then solve the dual problem in order to obtain the solution of
the primal. First, consider maximization (\ref{eq:g_dual}). As
expected, the maximization can not be solved in closed form.
However, as the feasible set (\textit{i.e.},
$\bm{A}_1,\cdots,\bm{A}_N\succeq0$) is the cartesian product of
convex sets, then a block coordinate ascent
algorithm\footnote{Also known as \textit{Non-Linear Gauss-Seidel}
Algorithm \cite[Section II-C]{PALOMAR06}.} can be used to search
for the maximum \cite[Section 2.7]{BERTSEKAS}. The algorithm
iteratively optimizes the function with respect to one $\bm{A}_n$
while keeping the others fixed. It has been previously used to
\textit{e.g.,} solve the sum-rate problem of MIMO multiple access
channels with individual and sum-power constraint
\cite{YU04}\cite{YU03}. We define it for our problem as:
\begin{eqnarray}\label{eq:g_dual_Gaussseidel}
\bm{A}_n^{t+1} = \arg\max_{\bm{A}_n\succeq0}
\mathcal{L}\left(\bm{A}_1^{t+1},\cdots,\bm{A}_{n-1}^{t+1},\bm{A}_{n},\bm{A}_{n+1}^{t},\cdots,\bm{A}_N^{t},\lambda\right),
\end{eqnarray}where $t$ is the iteration index. As shown in
Theorem \ref{Prop3}, the maximization
(\ref{eq:g_dual_Gaussseidel}) is uniquely attained.

\begin{theorem}\label{Prop3} Let the optimization $\bm{A}_n^*=\arg\max_{\bm{A}_n\succeq0}
\mathcal{L}\left(\bm{A}_1,\cdots,\bm{A}_N,\lambda\right)$ and the
conditional covariance matrix (See Appendix
\ref{cova_single_user})
\begin{eqnarray}\label{cond_cova_algo2}
\bm{R}_{{\bm{Y}}_{n}|\bm{Y}_0,\hat{\bm{Y}}_{n}^c}=
\bm{H}_{s,n}\left(\bm{I}+\bm{Q}\left(\frac{1}{\sigma_r^2}\bm{H}_{s,0}^\dagger\bm{H}_{s,0}+\sum_{p\neq
n}
\bm{H}_{s,p}^\dagger\left(\bm{A}_p\sigma_r^2\bm{I}+\bm{I}\right)^{-1}\bm{A}_p\bm{H}_{s,p}\right)\right)^{-1}\bm{Q}\bm{H}_{s,n}^\dagger
+ \sigma_r^2\bm{I}
\end{eqnarray} with eigen-decomposition
$\bm{R}_{{\bm{Y}}_{n}|\bm{Y}_0,\hat{\bm{Y}}_{n}^c}=\bm{U}_{n}\bm{S}\bm{U}_{n}^\dagger$.
The optimization is uniquely attained at
$\bm{A}_n^{*}=\bm{U}_{n}\bm{\eta}\bm{U}_{n}^\dagger$, where
\begin{eqnarray}\label{eq:prop3_eq1}
\eta_j=\left[\frac{1}{\lambda}\left(\frac{1}{\sigma_r^2}-\frac{1}{s_j}\right)-\frac{1}{\sigma_r^2}\right]^+,\
\ j=1,\cdots,N_n.
\end{eqnarray}
\end{theorem}
\begin{proof} See Appendix \ref{append:prop3} for the proof.\end{proof}

Function
$\mathcal{L}\left(\bm{A}_1,\cdots,\bm{A}_N,\lambda\right)$ is
continuously differentiable, and the maximization
(\ref{eq:g_dual_Gaussseidel}) is uniquely attained. Hence, the
limit point of the sequence
$\left\{\bm{A}_1^t,\cdots,\bm{A}_N^t\right\}$ is proven to
converge to a local maximum \cite[Proposition 2.7.1]{BERTSEKAS}.
To demonstrate convergence to the global maximum, it is necessary
to show that the mapping
$T\left(\bm{A}_1,\cdots,\bm{A}_N\right)=\left[\bm{A}_1+\gamma\nabla_{\bm{A}_1}\mathcal{L},\cdots,\bm{A}_N+\gamma\nabla_{\bm{A}_N}\mathcal{L}\right]$
is a block contraction\footnote{See \cite[Section
3.1.2]{BERTSEKAS_2} for the definition of block-contraction.} for
some $\gamma$ \cite[Proposition 3.10]{BERTSEKAS_2}. Unfortunately,
we were not able to demonstrate the contraction property on the
Lagrangian, although simulation results suggest global convergence
of our algorithm always.

Once obtained $g\left(\lambda\right)$ through the Gauss-Seidel
Algorithm\footnote{Assume hereafter that the algorithm has
converged to the global maximum of
$\mathcal{L}\left(\bm{A}_1,\cdots,\bm{A}_N,\lambda\right)$.}, it
remains to minimize it on $\lambda\geq0$. First, recall that
$g\left(\lambda\right)$ is a convex function, defined as the
pointwise maximum of a family of affine functions \cite{BOYD}.
Hence, to minimize it, we may use a subgradient approach as
\textit{e.g.}, that proposed by Yu in \cite{YU03}. The subgradient
search consists on following search direction $-h$ such that
\begin{eqnarray}\label{subgradef}
\frac{g\left(\lambda'\right)-g\left(\lambda\right)}{\lambda'-\lambda}\geq
h \qquad \forall \lambda'.
\end{eqnarray}Such a search is proven to converge to the global minimum for diminishing step-size
rules \cite[Section II-B]{PALOMAR06}. Considering the definition
of $g\left(\lambda\right)$, the following $h$ satisfies
(\ref{subgradef}):
\begin{eqnarray}\label{h_search}
h=\mathrm{R}-\log\det\left(\bm{I}+\textrm{diag}\left(\bm{A}_1,\cdots,\bm{A}_N\right)\bm{R}_{\bm{Y}_{1:N}|\bm{Y}_0}\right).
\end{eqnarray} Therefore, it is used to search for the optimum
$\lambda$ as:
\begin{eqnarray}\label{search}
\textrm{increase} \ \lambda \ \textrm{if} \ \ h\leq
0\ \ \textrm{or} \ \ 
\textrm{decrease} \ \lambda \ \textrm{if} \ \ h\geq 0.
\end{eqnarray}
Consider now $\lambda^0=1$ as the initial value of the Lagrange
multiplier. For such a multiplier, the optimum solution of
(\ref{eq:g_dual}) is
$\left\{\bm{A}_1^*.\cdots,\bm{A}_N^*\right\}=\bm{0}$ and the
subgradient (\ref{h_search}) is $h=\mathrm{R}$ (See Appendix
\ref{landamayor1}). Hence, following (\ref{search}), the optimum
value of $\lambda$ is strictly lower than one. Algorithm
\ref{algo1} takes all this into account in order to solve the dual
problem, hence solving the primal too. As mentioned, we can only
claim convergence of the algorithm to a local maximum.

\begin{algorithm}[h!]
\caption{Multiple-BSs dual problem}\label{algo1}
\begin{algorithmic}[1]
\medskip
\STATE Initialize $\lambda_{\mathrm{min}}=0$ and
$\lambda_{\mathrm{max}}=1$ \REPEAT \STATE $\lambda =
\frac{\lambda_{\mathrm{max}}-\lambda_{\mathrm{min}}}{2}$ \STATE
Obtain $\left\{\bm{A}_1^*,\cdots,\bm{A}_N^*\right\}= \arg\max
\mathcal{L}\left(\bm{A}_1,\cdots,\bm{A}_N,\lambda\right)$ from
Algorithm \ref{algo2} \STATE Evaluate $h $ as in (\ref{h_search}).
\STATE if $h\leq 0$, then $\lambda_{\mathrm{min}}=\lambda$, else
$\lambda_{\mathrm{max}}=\lambda$
\UNTIL{$\lambda_{\mathrm{max}}-\lambda_{\mathrm{min}}\leq
\epsilon$} \STATE
$\left\{\bm{\Phi}_1^*,\cdots,\bm{\Phi}_N^{*}\right\}=\left\{\left(\bm{A}_1^*\right)^{-1},\cdots,\left(\bm{A}_N^{*}\right)^{-1}\right\}$
\end{algorithmic}
\end{algorithm}

\vspace{-5mm}
\begin{algorithm}[h!]
\caption{Non-linear Gauss-Seidel to obtain
$g\left(\lambda\right)$}\label{algo2}
\begin{algorithmic}[1]
\medskip
\STATE Initialize $\bm{A}_n^0=\bm{0}$, $n=1,\cdots,N$ and $t=0$
\REPEAT \FOR{$n$ = 1 to $N$} \STATE Compute
$\bm{R}_{{\bm{Y}}_{n}|\bm{Y}_0,\hat{\bm{Y}}_{n}^c}\left(\bm{A}_1^{t+1},\cdots,\bm{A}_{n-1}^{t+1},\bm{A}_{n+1}^{t},\cdots,\bm{A}_N^t\right)$
from (\ref{cond_cova_algo2}). \STATE Take its eigen-decomposition
$\bm{U}_n\bm{S}\bm{U}_n^\dagger$ and compute $\bm{\eta}$ as in
(\ref{eq:prop3_eq1}). \STATE Update
$\bm{A}_n^{t+1}=\bm{U}_{n}\bm{\eta}\bm{U}_{n}^\dagger$. \ENDFOR
\STATE $t=t+1$ \UNTIL{The sequence converges
$\left\{\bm{A}_1^{t},\cdots,\bm{A}_N^{t}\right\}\rightarrow\left\{\bm{A}_1^*,\cdots,\bm{A}_N^*\right\}$}
\STATE Return $\left\{\bm{A}_1^*,\cdots,\bm{A}_N^*\right\}$
\end{algorithmic}
\end{algorithm}

\subsection{Practical Implementation} In the network, Distributed Wyner-Ziv compression can be practically implemented
using a simple Successive Wyner-Ziv (S-WZ) approach
\cite{BERGER07}\cite[Theorem 3]{DRAPER04}. To describe it, let us
recall that the optimum compression noises
$\bm{\Phi}_1^*,\cdots,\bm{\Phi}_N^*$ are obtained from Algorithm
\ref{algo1}, and let $\pi\left(\cdot\right)$ be a given
permutation on $\left\{1,\cdots,N\right\}$. For such a
permutation, the S-WZ coding is defined as follows:
\begin{itemize}
\item \textit{Parallel Compression}:
$\mathrm{BS}_{\pi\left(1\right)}$ compresses its received vector
using a single-source Wyner-Ziv code with decoder side information
$\bm{Y}_0$ (following Proposition \ref{Def_wyner_single_source}),
at a compression rate
\begin{eqnarray}
\rho_{\pi\left(1\right)}&=&I\left(\bm{Y}_{\pi\left(1\right)};\hat{\bm{Y}}_{\pi\left(1\right)}|\bm{Y}_0\right)\nonumber\\
&=&\log\det\left(\bm{I}+\left(\bm{\Phi}_{\pi\left(1\right)}^*\right)^{-1}\bm{R}_{\bm{Y}_{\pi\left(1\right)}|\bm{Y}_0}\right).
\end{eqnarray} The conditional covariance is calculated in
(\ref{app:cond_cova_singleuser_YnYo_final}). In parallel,
$\mathrm{BS}_{\pi\left(n\right)}$ $n>1$, compresses its signal
using a single-source Wyner-Ziv code with decoder side information
$\left(\bm{Y}_0,\hat{\bm{Y}}_{\pi\left(1:n-1\right)}\right)$, at a
rate
\begin{eqnarray}\label{eq:S-WZrate}
\rho_{\pi\left(n\right)}&=&I\left(\bm{Y}_{\pi\left(n\right)};\hat{\bm{Y}}_{\pi\left(n\right)}|\bm{Y}_0,\hat{\bm{Y}}_{\pi\left(1:n-1\right)}\right)\nonumber\\
&=&\log\det\left(\bm{I}+\left(\bm{\Phi}_{\pi\left(n\right)}^*\right)^{-1}\bm{R}_{\bm{Y}_{\pi\left(1\right)}|\bm{Y}_0,\hat{\bm{Y}}_{\pi\left(1:n-1\right)}}\right).
\end{eqnarray}In this case, the conditional covariance can be calculated
from (\ref{app:cond_cova_singleuser_YnYo_yg_final}). \item
\textit{Successive Decompression}: $\mathrm{BS}_0$ first recovers
the codeword $\hat{\bm{Y}}_{\pi\left(1\right)}$ using side
information $\bm{Y}_0$; later, it successively recovers codewords
$\hat{\bm{Y}}_{\pi\left(n\right)}$, $n>1$, using
$\bm{Y}_0,\hat{\bm{Y}}_{\pi\left(1:n-1\right)}$ as side
information.
\end{itemize}

It is easy to check the optimality of the S-WZ coding:
\begin{eqnarray}\label{eq:S-WZ}
\sum_{n=1}^N\rho_{\pi\left(n\right)} &=& \sum_{n=1}^N
I\left(\bm{Y}_{\pi\left(n\right)};\hat{\bm{Y}}_{\pi\left(n\right)}|\bm{Y}_0,\hat{\bm{Y}}_{\pi\left(1:n-1\right)}\right)\\
&=&\sum_{n=1}^N
I\left(\bm{Y}_{1:N};\hat{\bm{Y}}_{\pi\left(n\right)}|\bm{Y}_0,\hat{\bm{Y}}_{\pi\left(1:n-1\right)}\right)\nonumber\\
&=&I\left(\bm{Y}_{1:N};\hat{\bm{Y}}_{1:N}|\bm{Y}_0\right)\nonumber\\
&=&\mathrm{R}\nonumber.
\end{eqnarray}Second equality comes from the Markov chain
in Proposition \ref{Def_berger_tung}, and third from the chain
rule for mutual information; The fourth follows from the fact that
$\bm{\Phi}_1^*,\cdots,\bm{\Phi}_N^*$ satisfy the constraint
(\ref{eq:rate_prop1}) with equality. Unfortunately, transform
coding is not (generally) optimum for S-WZ with $N>1$, since the
eigenvectors of
$\bm{\Phi}_{\pi\left(n\right)}^*=\bm{U}_n\bm{\eta}^{-1}\bm{U}_n^\dagger$,
and those of
$\bm{R}_{\bm{Y}_{\pi\left(1\right)}|\bm{Y}_0,\hat{\bm{Y}}_{\pi\left(1:n-1\right)}}=\bm{V}_n\bm{S}\bm{V}_n^\dagger$
does necessarily match.

\section{The Multiple User Scenario}\label{sec:multiuser}

In previous sections, we considered a single user within the
network. To complement the analysis, we study hereafter multiple
(\textit{i.e.}, two) senders transmitting simultaneously. The
users, $s_1$ and $s_2$, transmit two independent messages
$\omega_u\in\left\{1,\cdots,2^{nR_{u}}\right\}$, $u=1,2$,  mapped
onto codewords $\bm{X}_{u}^n$, $u=1,2$, respectively. Codewords
are drawn \textit{i.i.d.} from random vectors
$\bm{X}_{u}\sim\mathcal{CN}\left(\bm{0,\bm{Q}_u}\right) $, $u=1,2$
and are not subject to optimization. Hence, now, the BSs receive:
\begin{eqnarray}\label{eq:signal_model}
\bm{Y}_i^n = \sum_{u=1}^2 \bm{H}_{u,i}\bm{X}_{u}^n+ \bm{Z}_i^n,\ \
i=0,\cdots,N.
\end{eqnarray}Here, $\bm{H}_{u,i}$ is the MIMO channel between
user $s_u$ and $\textrm{BS}_i$, and
$\bm{Z}_i\sim\mathcal{CN}\left(0,\sigma_r^2\bm{I}\right)$. As
previously, signals at $\mathrm{BS}_1,\cdots,\mathrm{BS}_N$ are
distributely compressed using a D-WZ code, and later sent to
$\mathrm{BS}_0$, which centralizes decoding. Using standard
arguments, the set $\mathcal{C}$ of transmission rates $R_{u}$,
$u=1,2$ at which messages $\omega_u$, $u=1,2$ can be reliably
decoded is \cite{COVER}\cite{SANDEROVICH08}:

\begin{eqnarray}\label{eq:MAC_capcityregion}
\mathcal{C}=\mathrm{coh}\left(\bigcup_{\underset{I\left(\bm{Y}_{1:N};\hat{\bm{Y}}_{1:N}|\bm{Y}_0\right)\leq
\mathrm{R}}{\prod_{i=1}^Np\left(\hat{\bm{Y}}_{i}|\bm{Y}_{i}\right):}}\left\{\left(R_{1},R_{2}\right):\begin{array}{c}
                                                                                     R_{1} \leq I\left(\bm{X}_{1};\bm{Y}_0,\hat{\bm{Y}}_{1:N}|\bm{X}_{2}\right) \\
                                                                                     R_{2} \leq I\left(\bm{X}_{2};\bm{Y}_0,\hat{\bm{Y}}_{1:N}|\bm{X}_{1}\right)\\
                                                                                     R_{1}+R_{2}\leq I\left(\bm{X}_{1},\bm{X}_{2};\bm{Y}_0,\hat{\bm{Y}}_{1:N}\right)
                                                                                   \end{array}
\right\}\right)
\end{eqnarray}

The union in (\ref{eq:MAC_capcityregion}) is explained by the fact
that compression codebooks might be arbitrary chosen at the BSs.
Notice that the boundary points of the region can be achieved
using superposition coding (SC) at the users, successive
interference cancellation (SIC) at the $\mathrm{BS}_0$, and
(optionally) \textit{time-sharing} (TS). Furthermore, as for the
single-user case, the optimum conditional distributions
$p\left(\hat{\bm{Y}}_{i}|\bm{Y}_{i}\right)$, $i=1,\cdots,N$ at the
boundary of the region can be proven to be
Gaussian\footnote{Recall that
$\bm{X}_{u}\sim\mathcal{CN}\left(\bm{0,\bm{Q}_u}\right) $,
$u=1,2$. We omit the proof due to space limitations.}. Therefore,
the union in (\ref{eq:MAC_capcityregion}) can be restricted to
compressed vectors of the form $\hat{\bm{Y}}_i = \bm{Y}_i +
\bm{Z}_i^c$, where
$\bm{Z}_i^c\sim\mathcal{CN}\left(0,\bm{\Phi}_i\right)$. That is:

\begin{eqnarray}\label{eq:MAC_capcityregion2}
 \mathcal{C}=\mathrm{coh}\left(\bigcup_{\underset{\in
c\left(\mathrm{R}\right)}{\bm{\Phi}_{1},\cdots,\bm{\Phi}_N}}\left\{\begin{array}{c}
 R_{1} \leq \log\det\left(\bm{I}+\frac{\bm{Q}_1}{\sigma_r^2}\bm{H}_{1,0}^\dagger\bm{H}_{1,0}+\bm{Q}_1\sum_{n=1}^N\bm{H}_{1,n}^\dagger\left(\sigma_r^2\bm{I}+\bm{\Phi}_n\right)^{-1}\bm{H}_{1,n}
\right)\\
R_{2}\leq
\log\det\left(\bm{I}+\frac{\bm{Q}_2}{\sigma_r^2}\bm{H}_{2,0}^\dagger\bm{H}_{2,0}+\bm{Q}_2\sum_{n=1}^N\bm{H}_{2,n}^\dagger\left(\sigma_r^2\bm{I}+\bm{\Phi}_n\right)^{-1}\bm{H}_{2,n}
\right)\\
R_{1}+R_{2}\leq
\log\det\left(\bm{I}+\frac{\bm{Q}}{\sigma_r^2}\bm{H}_{s,0}^\dagger\bm{H}_{s,0}+\bm{Q}\sum_{n=1}^N\bm{H}_{s,n}^\dagger\left(\sigma_r^2\bm{I}+\bm{\Phi}_n\right)^{-1}\bm{H}_{s,n}
\right)
                                                                                   \end{array}
\right\}\right)
\end{eqnarray}

Where $c\left(\mathrm{R}\right)=\left\{\bm{\Phi}_{1:N}:
\log\det\left(\bm{I}+\textrm{diag}\left(\bm{\Phi}_1^{-1},\cdots,\bm{\Phi}_N^{-1}\right)\bm{R}_{\bm{Y}_{1:N}|\bm{Y}_0}\right)\leq
\mathrm{R}\right\}$,
$\bm{Q}=\mathrm{diag}\left(\bm{Q}_1,\bm{Q_2}\right)$ and
$\bm{H}_{s,n} = \left[\bm{H}_{1,n},\ \bm{H}_{2,n}\right]$, for
$n=0,\cdots,N$. Covariance $\bm{R}_{\bm{Y}_{1:N}|\bm{Y}_0}$ is
calculated in Appendix \ref{cova_multiple_user}. To evaluate such
a region, we resort to the weighted sum-rate (WSR) optimization
\cite[Sec. III-C]{VERDÚ_CHENG93}. That is, we express
\begin{eqnarray}\label{eq:nyperplanes}
\mathcal{C}=\left\{\left(R_{1},R_{2}\right):\alpha R_{1} +
\left(1-\alpha\right)R_{2} \leq \mathcal{R}\left(\alpha\right),
\forall \alpha\in\left[0,1\right] \right\},
\end{eqnarray}with $\mathcal{R}\left(\alpha\right)$ the maximum
WSR, given weights $\alpha$ and $\left(1-\alpha\right)$ for user
$s_1$ and $s_2$, respectively. Such a WSR is achieved with
equality at the boundary of the region. Thus, it can be attained
considering SIC at $\mathrm{BS}_0$, which consists of first
decoding the user with lowest weight, considering second user as
interference. Later, once decoded the first user, the decoder
substracts its contribution to the received signal, and then
decodes the second user without interference.


\subsection{Useful Outer Regions}
Prior to solving the WSR optimization, we present two outer
regions on (\ref{eq:MAC_capcityregion2}).
\begin{outerregion} Rate region
(\ref{eq:MAC_capcityregion2}) is contained within the region
\begin{eqnarray}\label{eq:MAC_outerbound}
\begin{array}{c}
 R_{1} \leq \log\det\left(\bm{I}+\frac{\bm{Q}_1}{\sigma_r^2}\sum_{n=0}^N\bm{H}_{1,n}^\dagger
\bm{H}_{1,n}\right)\\
R_{2}\leq
\log\det\left(\bm{I}+\frac{\bm{Q}_2}{\sigma_r^2}\sum_{n=0}^N\bm{H}_{2,n}^\dagger
\bm{H}_{2,n}\right)\\
R_{1}+R_{2}\leq
\log\det\left(\bm{I}+\frac{\bm{Q}}{\sigma_r^2}\sum_{n=0}^N\bm{H}_{s,n}^\dagger
\bm{H}_{s,n}\right)
                                                                                   \end{array}
\end{eqnarray}
\end{outerregion}
\begin{remark} It is the capacity region when
$\bm{Y}_i$, $i=1,\cdots,N$ are available at $\mathrm{BS}_0$.
\end{remark}
\begin{outerregion}\label{upperbound4} The sum-rate satisfies
\begin{eqnarray}
R_{1}+R_{2} \leq
\log\det\left(\bm{I}+\frac{1}{\sigma_r^2}\bm{H}_{s,0}\bm{Q}\bm{H}_{s,0}^\dagger
\right) + \mathrm{R}.
\end{eqnarray}
\end{outerregion} \begin{proof} It is equivalent to the proof of upper bound
\ref{upperbound2}.\end{proof}


\subsection{Sum Rate Maximization}\label{sec:sum_rate}

The sum-rate of (\ref{eq:MAC_capcityregion2}) is identical to the
maximum transmission rate of a single user $s$ transmitting a
vector $\bm{X}_s=\left[\bm{X}_{1}^T,\bm{X}_{2}^T\right]^T$, with
equivalent channel $\bm{H}_{s,n} = \left[\bm{H}_{1,n},\
\bm{H}_{2,n}\right]$, $n=0,\cdots,N$. Hence, to maximize it we
resort to Algorithm \ref{algo1}.

\subsection{Weighted Sum Rate Maximization}

Let consider the WSR optimization with $\alpha > \frac{1}{2}$
(\textit{i.e.,} higher priority to user 1, which is decoded last
at the SIC). With such a decoding, the maximum rate of user 1 is
\begin{eqnarray}\label{eq:Rs1}
R_{1} =
I\left(\bm{X}_{1};\bm{Y}_0,\hat{\bm{Y}}_{1:N}|\bm{X}_{2}\right)\qquad \qquad \qquad \qquad \qquad \qquad \qquad \qquad \qquad \ \\
=
\log\det\left(\bm{I}+\frac{\bm{Q}_1}{\sigma_r^2}\bm{H}_{1,0}^\dagger\bm{H}_{1,0}+\bm{Q}_1\sum_{n=1}^N\bm{H}_{1,n}^\dagger\left(\sigma_r^2\bm{I}+\bm{\Phi}_n\right)^{-1}\bm{H}_{1,n}
\right). \nonumber
\end{eqnarray} On the other hand, the rate of
user 2, which is decoded first, follows:
\begin{eqnarray}\label{eq:Rs2}
R_{2} &=&
I\left(\bm{X}_{2};\bm{Y}_0,\hat{\bm{Y}}_{1:N}\right)\\
&=&
I\left(\bm{X}_{1},\bm{X}_{2};\bm{Y}_0,\hat{\bm{Y}}_{1:N}\right)-I\left(\bm{X}_{1};\bm{Y}_0,\hat{\bm{Y}}_{1:N}|\bm{X}_{2}\right)\nonumber\\
&=&\log\det\left(\bm{I}+\frac{\bm{Q}}{\sigma_r^2}\bm{H}_{s,0}^\dagger\bm{H}_{s,0}+\bm{Q}\sum_{n=1}^N\bm{H}_{s,n}^\dagger\left(\sigma_r^2\bm{I}+\bm{\Phi}_n\right)^{-1}\bm{H}_{s,n}
\right)-R_1,\nonumber
\end{eqnarray} where
$\bm{Q}=\mathrm{diag}\left(\bm{Q}_1,\bm{Q_2}\right)$ and
$\bm{H}_{s,n} = \left[\bm{H}_{1,n},\ \bm{H}_{2,n}\right]$. The
WSR, $\alpha R_1+\left(1-\alpha\right)R_2$, which has to be
maximized is convex on $\bm{\Phi}_{1},\cdots,\bm{\Phi}_N$. To make
it concave, we use the change the variables $\bm{\Phi}_n =
\bm{A}_n^{-1}$, $n=1,\cdots,N$. Then, plugging (\ref{eq:Rs1}) and
(\ref{eq:Rs2}) into (\ref{eq:nyperplanes}), the WSR optimization
turns into
\begin{eqnarray}\label{eq:wsr_concave}
\mathcal{R}\left(\alpha\right)&=&\max_{\bm{A}_{1},\cdots,\bm{A}_N}\
\alpha \cdot R_{1}+ \left(1-\alpha\right)\cdot R_{2}\qquad \qquad \ \ \ \ \\
&& \ \mathrm{s.t.} \ \
\log\det\left(\bm{I}+\textrm{diag}\left(\bm{A}_1,\cdots,\bm{A}_N\right)\bm{R}_{\bm{Y}_{1:N}|\bm{Y}_0}\right)\leq
\mathrm{R}\nonumber
\end{eqnarray} As previously, the constraint does not define a convex feasible
set. To solve the optimization, we follow the strategy presented
previously: first, we show that the optimization has zero duality
gap. Later, we propose an iterative algorithm that solves the dual
problem, thus solving the primal too.

\begin{lemma} The duality gap for the WSR optimization (\ref{eq:wsr_concave}) is
zero.
\end{lemma}
\begin{proof} Applying the time-sharing property in \cite[Theorem 1]{YU06} the zero-duality gap is demonstrated.\end{proof}

Let then solve the dual problem. The Lagrangian for optimization
(\ref{eq:wsr_concave}) is defined as:
\begin{eqnarray}\label{eq:WSRf_dual}
\mathcal{L}_\alpha\left(\bm{A}_1,\cdots,\bm{A}_n,\lambda\right) =
\alpha \cdot R_{1}+ \left(1-\alpha\right)\cdot R_{2} -
\lambda\cdot\left(\log\det\left(\bm{I}+\textrm{diag}\left(\bm{A}_1,\cdots,\bm{A}_N\right)\bm{R}_{\bm{Y}_{1:N}|\bm{Y}_0}\right)-
\mathrm{R}\right)
\end{eqnarray}The first step is to find the dual function \cite[Section 5]{BERTSEKAS}
\begin{eqnarray}\label{eq:WSRg_dual}
g_\alpha\left(\lambda\right)=\max_{\bm{A}_1,\cdots,\bm{A}_n\succeq0}\mathcal{L}_\alpha\left(\bm{A}_1,\cdots,\bm{A}_n,\lambda\right)
\end{eqnarray}

In previous sections, we showed that such an optimization can be
tackled using a block-coordinate algorithm. Unfortunately, now,
the maximization with respect to a single $\bm{A}_n$ cannot be
solved in closed-form, and is not clear to be uniquely attained.
Hence, to solve (\ref{eq:WSRg_dual}), we propose another
algorithm: the \textit{gradient projection method} (GP)
\cite[Section 2.3]{BERTSEKAS}. GP has been used to \textit{e.g.,}
compute transmit covariances for MIMO interference channels, and
the WSR of MIMO broadcast channels \cite[Section
IV-C]{YE03}\cite{LIU07}. It is defined as follows: let
(\ref{eq:WSRg_dual}), and consider the initial point
$\left\{\bm{A}_1^0,\cdots,\bm{A}_n^0\right\}\succeq0$. It
iteratively updates \cite[Section 2.3.1]{BERTSEKAS}:
\begin{eqnarray}\label{eq:GP}
\bm{A}_n^{t+1} =
\bm{A}_n^{t}+\gamma_t\left(\bar{\bm{A}}_n^{t}-\bm{A}_n^{t}\right),\
\ n=1,\cdots,N
\end{eqnarray}where $t$ is the iteration index and $0<\gamma_t\leq1$ is the step size. Also,
\begin{eqnarray}
\bar{\bm{A}}_n^{t} =
\left[\bm{A}_n^{t}+s_t\cdot\nabla_{\bm{A}_n}\mathcal{L}_\alpha\left(\lambda,\bm{A}_1^{t},\cdots,\bm{A}_N^{t}\right)\right]_{\succeq0},\
\ n=1,\cdots,N
\end{eqnarray}with $s_t\geq0$ an scalar and
$\nabla_{\bm{A}_n}\mathcal{L}_\alpha\left(\lambda,\bm{A}_1^{t},\cdots,\bm{A}_N^{t}\right)$
the gradient of $\mathcal{L}_\alpha\left(\cdot\right)$ with
respect to $\bm{A}_n$, evaluated at
$\bm{A}_1^{t},\cdots,\bm{A}_N^{t}$. Finally,
$\left[\cdot\right]_{\succeq0}$ denotes the projection (with
respect to the Frobenius norm) onto the cone of positive
semidefinite matrices. Whenever $\gamma_t$ and $s_t$ are chosen
appropriately, the sequence
$\left\{\bm{A}_1^t,\cdots,\bm{A}_n^t\right\}$ is proven to
converge to a local maximum of (\ref{eq:WSRg_dual})
\cite[Proposition 2.2.1]{BERTSEKAS}. (For global convergence to
hold, the contraction property must be satisfied. Unfortunately,
we were not able to prove this property for our optimization). In
order to make the algorithm work for the problem, we need to:
\textit{i}) compute the projection of a Hermitian matrix $\bm{S}$,
with eigen-decomposition $\bm{S}=\bm{U}\bm{\eta}\bm{U}^\dagger$,
onto the cone of positive semidefinite matrices. It is equal to
\cite[Theorem 2.1]{MALICK06}:
\begin{eqnarray}\label{projection}
\left[\bm{S}\right]_{\succeq0}=
\bm{U}\mathrm{diag}\left(\max\left\{\eta_1,0\right\},\cdots,\max\left\{\eta_m,0\right\}\right)\bm{U}^\dagger.
\end{eqnarray}
\textit{ii}) Obtain the gradient of
$\mathcal{L}_\alpha\left(\cdot\right)$ with respect to a single
$\bm{A}_n$, which is twice the conjugate of the partial derivative
of the function with respect to such a matrix
\cite{MATRIXCOOKBOOK}:
\begin{eqnarray}\label{eq:grad_def}
\nabla_{\bm{A}_n}\mathcal{L}_\alpha\left(\bm{A}_{1:N},\lambda\right)=
2\left(\left[\frac{\partial
\mathcal{L}_\alpha\left(\bm{A}_{1:N},\lambda\right)}{\partial\bm{A}_n}\right]^T\right)^\dagger
\end{eqnarray}

The Lagrangian is defined in (\ref{eq:WSRf_dual}). To obtain its
partial derivative, we make use of (\ref{constraint_extended}):
\begin{eqnarray}\label{eq:grad_constraint}
\left[\frac{\partial
\log\det\left(\bm{I}+\textrm{diag}\left(\bm{A}_1,\cdots,\bm{A}_N\right)\bm{R}_{\bm{Y}_{1:N}|\bm{Y}_0}\right)}{\partial\bm{A}_n}\right]^T
=
\left[\frac{\partial\log\det\left(\bm{I}+\bm{A}_n\bm{R}_{\bm{Y}_{n}|\bm{Y}_0,\hat{\bm{Y}}_{n}^c}\right)}{\partial\bm{A}_n}\right]^T\\
=\bm{R}_{\bm{Y}_{n}|\bm{Y}_0,\hat{\bm{Y}}_{n}^c}\left(\bm{I}+\bm{A}_n\bm{R}_{\bm{Y}_{n}|\bm{Y}_0,\hat{\bm{Y}}_{n}^c}\right)^{-1}.\nonumber
\end{eqnarray}The conditional covariance is computed in Appendix
\ref{cova_multiple_user}. Furthermore, we can also derive that
\begin{eqnarray}\label{eq:grad_RS1}
\frac{\partial R_{1}}{\partial\bm{A}_n} &=& \frac{\partial
I\left(\bm{X}_{1};\bm{Y}_0,\hat{\bm{Y}}_{1:N}|\bm{X}_{2}\right)}{\partial\bm{A}_n}\\
&=& \frac{\partial
I\left(\bm{X}_{1};\hat{\bm{Y}}_{n}|\bm{X}_{2},\bm{Y}_0,\hat{\bm{Y}}_{n}^c\right)}{\partial\bm{A}_n}\nonumber
\end{eqnarray}where second equality follows from the chain rule for
mutual information and noting that
$I\left(\bm{X}_{1};\bm{Y}_0,\hat{\bm{Y}}_{n}^c|\bm{X}_{2}\right)$
does not depend on $\bm{A}_n$. The mutual information above is
evaluated as:
\begin{eqnarray}\label{eq:I_RS1}
I\left(\bm{X}_{1};\hat{\bm{Y}}_{n}|\bm{X}_{2},\bm{Y}_0,\hat{\bm{Y}}_{n}^c\right)
&=&H\left(\hat{\bm{Y}}_{n}|\bm{X}_{2},\bm{Y}_0,\hat{\bm{Y}}_{n}^c\right)-H\left(\hat{\bm{Y}}_{n}|\bm{X}_{1},\bm{X}_{2},\bm{Y}_0,\hat{\bm{Y}}_{n}^c\right)\\
&=&
\log\det\left(\bm{R}_{{\bm{Y}}_{n}|\bm{X}_{2},\bm{Y}_0,\hat{\bm{Y}}_{n}^c}+\bm{\Phi}_n\right)-\log\det\left(\sigma_r^2\bm{I}+\bm{\Phi}_n\right)\nonumber\\
&=&
\log\det\left(\bm{A}_n\bm{R}_{{\bm{Y}}_{n}|\bm{X}_{2},\bm{Y}_0,\hat{\bm{Y}}_{n}^c}+\bm{I}\right)-\log\det\left(\bm{A}_n\sigma_r^2+\bm{I}\right)\nonumber
\end{eqnarray}Last equality follows from
$\bm{\Phi}_n=\bm{A}_n^{-1}$, and
$\bm{R}_{{\bm{Y}}_{n}|\bm{X}_{2},\bm{Y}_0,\hat{\bm{Y}}_{n}^c}$ is
computed in Appendix \ref{cova_multiple_user}. Therefore, the
derivative of $R_{1}$ remains \cite{MATRIXCOOKBOOK}
\begin{eqnarray}\label{eq:grad_RS1_final}
\left[\frac{\partial R_{1}}{\partial\bm{A}_n}\right]^T =
\bm{R}_{{\bm{Y}}_{n}|\bm{X}_{2},\bm{Y}_0,\hat{\bm{Y}}_{n}^c}\left(\bm{A}_n\bm{R}_{{\bm{Y}}_{n}|\bm{X}_{2},\bm{Y}_0,\hat{\bm{Y}}_{n}^c}+\bm{I}\right)^{-1}-\sigma_r^2\left(\bm{A}_n\sigma_r^2+\bm{I}\right)^{-1}.
\end{eqnarray}
Equivalently, we can obtain for the derivative of $R_{2}$ that
\begin{eqnarray}\label{eq:grad_RS2}
\frac{\partial R_{2}}{\partial\bm{A}_n} &=& \frac{\partial
I\left(\bm{X}_{2};\bm{Y}_0,\hat{\bm{Y}}_{1:N}\right)}{\partial\bm{A}_n}\\
&=&\frac{\partial
I\left(\bm{X}_{2};\hat{\bm{Y}}_{n}|\bm{Y}_0,\hat{\bm{Y}}_{n}^c\right)}{\partial\bm{A}_n}\nonumber.
\end{eqnarray} Where we evaluate:
\begin{eqnarray}\label{eq:I_RS2}
I\left(\bm{X}_{2};\hat{\bm{Y}}_{n}|\bm{Y}_0,\hat{\bm{Y}}_{n}^c\right)
&=&H\left(\hat{\bm{Y}}_{n}|\bm{Y}_0,\hat{\bm{Y}}_{n}^c\right)-H\left(\hat{\bm{Y}}_{n}|\bm{X}_{2},\bm{Y}_0,\hat{\bm{Y}}_{n}^c\right)\\
&=&\log\det\left(\bm{A}_n\bm{R}_{{\bm{Y}}_{n}|\bm{Y}_0,\hat{\bm{Y}}_{n}^c}+\bm{I}\right)-\log\det\left(\bm{A}_n\bm{R}_{{\bm{Y}}_{n}|\bm{X}_{2},\bm{Y}_0,\hat{\bm{Y}}_{n}^c}+\bm{I}\right)\nonumber
\end{eqnarray}Conditional covariances are obtained in
Appendix \ref{cova_multiple_user}. The derivative of $R_{2}$ thus
remains:
\begin{eqnarray}\label{eq:grad_RS2_final}
\left[\frac{\partial R_{2}}{\partial\bm{A}_n}\right]^T
=\bm{R}_{{\bm{Y}}_{n}|\bm{Y}_0,\hat{\bm{Y}}_{n}^c}\left(\bm{A}_n\bm{R}_{{\bm{Y}}_{n}|\bm{Y}_0,\hat{\bm{Y}}_{n}^c}+\bm{I}\right)^{-1}-\bm{R}_{{\bm{Y}}_{n}|\bm{X}_{2},\bm{Y}_0,\hat{\bm{Y}}_{n}^c}\left(\bm{A}_n\bm{R}_{{\bm{Y}}_{n}|\bm{X}_{2},\bm{Y}_0,\hat{\bm{Y}}_{n}^c}+\bm{I}\right)^{-1}.\end{eqnarray}

Plugging (\ref{eq:grad_constraint}), (\ref{eq:grad_RS1_final}) and
(\ref{eq:grad_RS2_final}) into (\ref{eq:grad_def}) we obtain the
gradient of the function, which is used in the GP algorithm to
obtain $g_{\alpha}\left(\lambda\right)$. Notice that for $\alpha
\leq \frac{1}{2}$, the roles of users $s_1$ and $s_2$ are
interchanged, being user 1 decoded first. This roles would also
need to be interchanged in the computation of the gradients of
$R_{1}$ and $R_{2}$. Once obtained the dual function, we minimize
it to obtain:
\begin{eqnarray}
\mathcal{R}\left(\alpha\right)=\min_{\lambda\geq0}\
g_\alpha\left(\lambda\right).
\end{eqnarray}To solve this minimization, we use the subgradient
approach as in Section \ref{sec:multiple basestations}. Taking all
this into account we build up Algorithm \ref{algo3}. As for the
previous section, we can only claim local convergence.

\begin{algorithm}[h!]
\caption{Two-user WSR dual problem}\label{algo3}
\begin{algorithmic}[1]
\medskip
\STATE Initialize $\lambda_{\mathrm{min}}=0$ and
$\lambda_{\mathrm{max}}$ \REPEAT \STATE $\lambda =
\frac{\lambda_{\mathrm{max}}-\lambda_{\mathrm{min}}}{2}$ \STATE
Obtain $\left\{\bm{A}_1^*,\cdots,\bm{A}_N^*\right\}=
\arg\max\mathcal{L}_\alpha\left(\bm{A}_1,\cdots,\bm{A}_n,\lambda\right)$
from Algorithm \ref{algo4} \STATE Evaluate $h$ as in
(\ref{h_search}), where $\bm{R}_{\bm{Y}_{1:N}|\bm{Y}_0}$ follows
Appendix \ref{cova_multiple_user}. \STATE if $h\leq 0$, then
$\lambda_{\mathrm{min}}=\lambda$, else
$\lambda_{\mathrm{max}}=\lambda$
\UNTIL{$\lambda_{\mathrm{max}}-\lambda_{\mathrm{min}}\leq
\epsilon$} \STATE $\mathcal{R}\left(\alpha\right)=\alpha
R_{1}\left(\bm{A}_1^*,\cdots,\bm{A}_N^*\right)+\left(1-\alpha\right)R_{2}\left(\bm{A}_1^*,\cdots,\bm{A}_N^*\right)$.

\end{algorithmic}
\end{algorithm}

\begin{algorithm}[h!]
\caption{GP to obtain
$g_{\alpha}\left(\lambda\right)$}\label{algo4}
\begin{algorithmic}[1]
\medskip
\STATE Initialize $\bm{A}_n^0=\bm{0}$, $n=1,\cdots,N$ and $t=0$
\REPEAT \STATE Compute the gradient
$\bm{G}_n^{t}=\nabla_{\bm{A}_n}\mathcal{L}_\alpha\left(\lambda,\bm{A}_1^{t},\cdots,\bm{A}_N^{t}\right)$,
$n=1,\cdots,N$ from (\ref{eq:grad_def}).\STATE  Choose appropriate
$s_t$ \STATE Set $\hat{\bm{A}}_n^{t} =
\bm{A}_n^{t}+s_t\cdot\bm{G}_n^{t}$. Calculate $\hat{\bm{A}}_n^{t}=
\bm{U}_n\bm{\eta}\bm{U}_n^\dagger$. Then, $\bar{\bm{A}}_n^{t}=
\bm{U}_n\max\left\{\bm{\eta},0\right\}\bm{U}_n^\dagger$,
$n=1,\cdots,N$. \STATE Choose appropriate $\gamma_t$ \STATE Update
$\bm{A}_n^{t+1} =
\bm{A}_n^{t}+\gamma_t\left(\bar{\bm{A}}_n^{t}-\bm{A}_n^{t}\right)$,
$n=1,\cdots,N$ \STATE $t=t+1$ \UNTIL{The sequence converges
$\left\{\bm{A}_1^{t},\cdots,\bm{A}_N^{t}\right\}\rightarrow\left\{\bm{A}_1^*,\cdots,\bm{A}_N^*\right\}$}
\STATE Return $\left\{\bm{A}_1^*,\cdots,\bm{A}_N^*\right\}$
\end{algorithmic}
\end{algorithm}
\vspace{-7mm}
%


\section{Numerical Results}\label{sec:numericalresults}

We evaluate the performance of D-WZ coding within a
single-frequency network composed of a central base station
$\mathrm{BS}_0$ plus its first tier of six cells. The radius of
each cell is 700 m, and BSs have all three receive antennas. On
the other hand, users have two antennas, are located at the edge
of the central cell and transmit isotropically, \textit{i.e.},
$\bm{Q}_i=\frac{\mathrm{P}_{TX}}{2}\bm{I}$. Transmitted power is
set to 23 dBm, and wireless channels are simulated taking into
account path loss, log-normal shadowing and Rayleigh fading.
Specifically, fading is assumed \textit{i.i.d.} among antennas,
and shadowing uncorrelated among BSs. Two propagation scenarios
are studied: \textit{i}) Line-of-sight (LOS), with path-loss
exponent $\alpha=2.6$ and shadowing standard deviation $\sigma=4$
dB. \textit{ii}) Non Line-of-sight (N-LOS), with $\alpha=4.05$ and
$\sigma=10$ dB.


Fig. \ref{fig:fig4} plots the cumulative density function (cdf) of
the uplink rate\footnote{The user is assumed to transmit at $1$
Mbaud, \textit{i.e.,} 1 Msymb/s.} for a single-user network,
considering different values of the backhaul rate $\mathrm{R}$.
Particularly, Fig. \ref{fig:fig4a} depicts results for LOS
propagation,  and shows gains up to 6 Mbit/s @ 5\% outage, with
$\mathrm{R} = 15$ Mbit/s. It is clearly shown that BSs cooperation
becomes more remarkable for lower outage probabilities. On the
other hand, Fig. \ref{fig:fig4b} shows results for N-LOS
propagation, where rate gains are reduced. In this case,
cooperation becomes more convenient for higher outages, showing
that @ 50\% outage, three-fold gains arise with 15 Mbit/s of
backhaul.

Fig \ref{fig:fig6} plots the uplink rate of a single-user network
with $\mathrm{R}= 7$ Mbit/s, for different number $N$ of
cooperative BSs. First, Fig. \ref{fig:fig6a} depicts the cdf of
the user's rate under LOS propagation conditions. We notice that @
5\% outage, with only 1 cooperative BS, a rate gain of 2 Mbit/s is
obtained with respect to the non-cooperative case. However, when
increasing the number of cooperative BSs to 6, only an additional
rate gain of 2 Mbit/s is obtained. That is, the impact of
introducing new cooperative BSs in the system diminishes as the
network grows. Again, cooperation is more useful for low outages.
On the other hand, Fig. \ref{fig:fig6b} depicts results for N-LOS
propagation. It can be shown that, @ 50\% outage, the rate is
doubled from 1 cooperative BS to 6 cooperative BS. This fact
highlights the relevant role of macro-diversity on N-LOS
conditions, which are most common ones on urban cellular networks.
Next, Fig. \ref{fig:fig7} compares the rate performance of our
D-WZ approach with respect to that of \textit{Quantization}
\cite{MARSCH07}, assuming LOS propagation. We consider a simple
network with two BSs: $\mathrm{BS}_0$ and $ \mathrm{BS}_1$, and
plot its outage capacity with D-WZ and with uniform quantization,
respectively. Both are normalized with respect to the outage
capacity with infinite backhaul and computed at a probability of
outage of $10^{-2}$. Results show significant gains, of up to
12\%, for low backhaul rates, and hihglights the fact that D-WZ
requires half of backhaul rate than Quantization to converge to
the $\infty$ backhaul capacity.

Fig \ref{fig:fig8} depicts the expected sum-rate\footnote{The
expected sum-rate is obtained by averaging the sum-rate of the
system over the user's channels.}  of the multi-user setup versus
the total number of users. Results are shown for different values
of the backhaul rate. Although the sum-rate analysis (see Sec.
\ref{sec:sum_rate}) was carried out for two users only, the
extension to $U>2$ is straightforward. Fig \ref{fig:fig8a} depicts
the sum-rate for LOS propagation. We first notice that the sum
rate with $\infty$ backhaul capacity (\textit{i.e.,} outer region
1) is far away from the sum-rate with D-WZ compression. This is
explained by means of outer region 2: the sum-rate of the system
is constrained by the available rate at the backhaul network. On
the other hand, for N-LOS propagation (Fig. \ref{fig:fig8b}),
upper bound \ref{upperbound4} is not reached. Indeed, for less
than 5 users, the expected sum-rate with only $\mathrm{R} = 15$
Mbit/s of backhaul is almost identical to that of $\mathrm{R} =
\infty$. Therefore, for practical number of transmitters, the full
rate gain due to macro-diversity is obtained via D-WZ compression.
Finally, Fig. \ref{fig:fig10} and Fig. \ref{fig:fig11} depict the
rate region of a 2-user network, with and without LOS
respectively, for different values of the Backhaul rate
$\mathrm{R}$. It is clearly shown that the region is significantly
enlarged with only 5 Mbit/s of backhaul rate.

\section{Conclusions}

We studied distributed compression for the uplink of a coordinated
cellular network with $N+1$ multi-antenna BSs. Considering a
constrained backhaul of limited capacity $\mathrm{R}$, base
stations $\mathrm{BS}_1,\cdots,\mathrm{BS}_N$ distributely
compress their received signal using a Distributed Wyner-Ziv code.
The compressed vectors are sent to $\mathrm{BS}_0$, which
centralizes user's decoding. Considering single and multiple users
within the network, respectively, the D-WZ scheme has been
optimized using the users' rate as the performance metric.

\appendices

\section{Conditional Covariances}\label{appen:covariances}

We derive here conditional covariances used throughout the paper.
(See supporting material)

\subsection{The single user case}\label{cova_single_user}
\begin{eqnarray}\label{app:cond_cova_singleuser_YnYo_final}
\bm{R}_{\bm{Y}_n|\bm{Y}_0}
=\bm{H}_{s,n}\left(\bm{I}+\frac{\bm{Q}}{\sigma_r^2}\bm{H}_{s,0}^\dagger\bm{H}_{s,0}\right)^{-1}\bm{Q}\bm{H}_{s,n}^\dagger+\sigma_r^2\bm{I},\
\ n=1,\cdots,N.
\end{eqnarray}
\begin{eqnarray}\label{app:cond_cova_singleuser_Y1:NYo_final}
\bm{R}_{\bm{Y}_{1:N}|\bm{Y}_0} =\left[\begin{array}{c}
         \bm{H}_{s,1} \\
         \vdots \\
         \bm{H}_{s,N}
       \end{array}
\right]\left(\bm{I}+\frac{\bm{Q}}{\sigma_r^2}\bm{H}_{s,0}^\dagger\bm{H}_{s,0}\right)^{-1}\bm{Q}\left[\begin{array}{c}
         \bm{H}_{s,1} \\
         \vdots \\
         \bm{H}_{s,N}
       \end{array}
\right]^\dagger+\sigma_r^2\bm{I}.
\end{eqnarray}
\begin{eqnarray}\label{app:cond_cova_singleuser_YnYo_ync_final}
\bm{R}_{\bm{Y}_n|\bm{Y}_0,\hat{\bm{Y}}_n^c}
=\bm{H}_{s,n}\left(\bm{I}+\frac{\bm{Q}}{\sigma_r^2}\bm{H}_{s,0}^\dagger\bm{H}_{s,0}+\sum_{j\neq
n}\bm{Q}\bm{H}_{s,j}^\dagger\left(\sigma_r^2\bm{I}+\bm{\Phi}_j\right)^{-1}\bm{H}_{s,j}\right)^{-1}\bm{Q}\bm{H}_{s,n}^\dagger+\sigma_r^2\bm{I}.
\end{eqnarray}
\begin{eqnarray}\label{app:cond_cova_singleuser_YnYo_yg_final}
\bm{R}_{\bm{Y}_n|\bm{Y}_0,\hat{\bm{Y}}_\mathcal{G}}
=\bm{H}_{s,n}\left(\bm{I}+\frac{\bm{Q}}{\sigma_r^2}\bm{H}_{s,0}^\dagger\bm{H}_{s,0}+\sum_{j\in
\mathcal{G}}\bm{Q}\bm{H}_{s,j}^\dagger\left(\sigma_r^2\bm{I}+\bm{\Phi}_j\right)^{-1}\bm{H}_{s,j}\right)^{-1}\bm{Q}\bm{H}_{s,n}^\dagger+\sigma_r^2\bm{I}.
\end{eqnarray}
\subsection{The multiuser case}\label{cova_multiple_user}

Define $\bm{H}_{s,n}=\left[\bm{H}_{1,n},\bm{H}_{2,n}\right]$ and
$\bm{Q}=\mathrm{diag}\left(\bm{Q}_1,\bm{Q}_2\right)$. Then,
Conditional covariances $\bm{R}_{\bm{Y}_n|\bm{Y}_0}$,
$\bm{R}_{\bm{Y}_{1:N}|\bm{Y}_0}$
$\bm{R}_{\bm{Y}_n|\bm{Y}_0,\hat{\bm{Y}}_n^c}$ and
$\bm{R}_{\bm{Y}_n|\bm{Y}_0,\hat{\bm{Y}}_\mathcal{G}}$ follow
Subsection \ref{cova_single_user}. Furthermore, let
$i,j\in\left\{1,2\right\}$ with $j\neq i$, then:
\begin{eqnarray}\label{app:cond_cova_multiuser_YnXsiYo_ync_final}
\bm{R}_{\bm{Y}_n|\bm{X}_{i},\bm{Y}_0,\hat{\bm{Y}}_n^c} =
\bm{H}_{j,n}\left(\bm{I}+\frac{\bm{Q}_j}{\sigma_r^2}\bm{H}_{j,0}^\dagger\bm{H}_{j,0}+\sum_{p\neq
n}\bm{Q}_j\bm{H}_{j,p}^\dagger\left(\sigma_r^2\bm{I}+\bm{\Phi}_p\right)^{-1}\bm{H}_{j,p}\right)^{-1}\bm{Q}_j\bm{H}_{j,n}^\dagger+\sigma_r^2\bm{I}
\end{eqnarray}

\section{Proof of Proposition \ref{prop1}} \label{appen:prop1}
Let the chain rule for mutual information:
\begin{eqnarray}\label{proof:prop1_0}
I\left(\bm{X}_s;\bm{Y}_0,\hat{\bm{Y}}_{1:N}\right) =
I\left(\bm{X}_s;\bm{Y}_0\right)+I\left(\bm{X}_s;\hat{\bm{Y}}_{1:N}|\bm{Y}_0\right).
\end{eqnarray}Also, let expand the constraint to obtain:
\begin{eqnarray}\label{proof:prop1_1}
I\left(\bm{Y}_{1:N};\hat{\bm{Y}}_{1:N}|\bm{Y}_0\right) &=&
H\left(\hat{\bm{Y}}_{1:N}|\bm{Y}_0\right) -
H\left(\hat{\bm{Y}}_{1:N}|\bm{Y}_0,\bm{Y}_{1:N}\right)\nonumber\\
&=& I\left(\bm{X}_s;\hat{\bm{Y}}_{1:N}|\bm{Y}_0\right) +
H\left(\hat{\bm{Y}}_{1:N}|\bm{Y}_0,\bm{X}_s\right) -
H\left(\hat{\bm{Y}}_{1:N}|\bm{Y}_0,\bm{Y}_{1:N}\right).
\end{eqnarray}
Given the Markov chain in Theorem \ref{Def_berger_tung}:
$H\left(\hat{\bm{Y}}_{1:N}|\bm{Y}_0,\bm{Y}_{1:N}\right)=H\left(\hat{\bm{Y}}_{1:N}|\bm{Y}_0,\bm{X}_s,\bm{Y}_{1:N}\right)$,
which plugged into (\ref{proof:prop1_1}):
\begin{eqnarray}\label{proof:prop1_2}
I\left(\bm{Y}_{1:N};\hat{\bm{Y}}_{1:N}|\bm{Y}_0\right) =
I\left(\bm{X}_s;\hat{\bm{Y}}_{1:N}|\bm{Y}_0\right) +
I\left(\bm{Y}_{1:N};\hat{\bm{Y}}_{1:N}|\bm{Y}_0,\bm{X}_s\right).
\end{eqnarray}
Let now $\mathcal{P}$ be the feasible set of conditional
probabilities
$\prod_{i=1}^Np\left(\hat{\bm{Y}}_{i}|\bm{Y}_{i}\right)$,
\textit{i.e.}, the set for which
$I\left(\bm{Y}_{1:N};\hat{\bm{Y}}_{1:N}|\bm{Y}_0\right)\leq
\mathrm{R}$. Hence, making use of (\ref{proof:prop1_2}), the
feasible set satisfies:
\begin{eqnarray}\label{proof:prop1_3}
I\left(\bm{X}_s;\hat{\bm{Y}}_{1:N}|\bm{Y}_0\right) \leq \mathrm{R} -
I\left(\bm{Y}_{1:N};\hat{\bm{Y}}_{1:N}|\bm{Y}_0,\bm{X}_s\right).
\end{eqnarray} Introducing (\ref{proof:prop1_3}) into
(\ref{proof:prop1_0}), we derive that for the feasible set:
\begin{eqnarray}\label{proof:prop1_4}
I\left(\bm{X}_s;\bm{Y}_0,\hat{\bm{Y}}_{1:N}\right) \leq
I\left(\bm{X}_s;\bm{Y}_0\right)+\mathrm{R}-I\left(\bm{Y}_{1:N};\hat{\bm{Y}}_{1:N}|\bm{Y}_0,\bm{X}_s\right).
\end{eqnarray}Now, notice that
$I\left(\bm{Y}_{1:N};\hat{\bm{Y}}_{1:N}|\bm{Y}_0,\bm{X}_s\right)=I\left(\bm{Z}_{1:N};\hat{\bm{Y}}_{1:N}|\bm{Y}_0,\bm{X}_s\right)$
where $\bm{Z}_i$ is the AWGN at the $\textrm{BS}_i$. This mutual
information is minimized in $\mathcal{P}$ for
$p\left(\hat{\bm{Y}}_{1:N}\right)$ Gaussian. Therefore,
$I\left(\bm{X}_s;\bm{Y}_0,\hat{\bm{Y}}_{1:N}\right)$ in
(\ref{proof:prop1_4}) is maximum in $\mathcal{P}$ for Gaussian
distributed vectors $\hat{\bm{Y}}_{1:N}$, specifically those
satisfying
$I\left(\bm{Y}_{1:N};\hat{\bm{Y}}_{1:N}|\bm{Y}_0\right)=
\mathrm{R}$ (\textit{i.e.,} those for which equality holds in
(\ref{proof:prop1_4}) and (\ref{proof:prop1_3})). As mentioned,
the received vectors $\bm{Y}_i$ are also Gaussian. Therefore, at
the optimum, $\hat{\bm{Y}}_i$ and $\bm{Y}_i$ are jointly Gaussian,
so we can write $\hat{\bm{Y}}_i = \bm{M}\bm{Y}_i +
\hat{\bm{Z}}_i^c$ with $\bm{M}$ a constant matrix and
$\hat{\bm{Z}}_i^c$ an independent Gaussian vector. However, as the
multiplication by a matrix does not affect mutual information, we
can state that vectors $\hat{\bm{Y}}_i = \bm{Y}_i + {\bm{Z}}_i^c$
are also optimal, with
${\bm{Z}}_c^i\sim\mathcal{CN}\left(\bm{0},\bm{\Phi}_i\right)$.
Using this relationship, we evaluate
\begin{eqnarray}\label{app:rate_prop1}
I\left(\bm{X}_s;\bm{Y}_0,\hat{\bm{Y}}_{1:N}\right)=
\log\det\left(\bm{I}+\frac{\bm{Q}}{\sigma_r^2}\bm{H}_{s,0}^\dagger\bm{H}_{s,0}+\bm{Q}\sum_{n=1}^N\bm{H}_{s,n}^\dagger\left(\sigma_r^2\bm{I}+\bm{\Phi}_n\right)^{-1}\bm{H}_{s,n}
\right)
\end{eqnarray} Furthermore, we can also obtain:
\begin{eqnarray}\label{app:rate_prop2}
I\left(\bm{Y}_{1:N};\hat{\bm{Y}}_{1:N}|\bm{Y}_0\right) &=&
H\left(\hat{\bm{Y}}_{1:N}|\bm{Y}_0\right)-
H\left(\hat{\bm{Y}}_{1:N}|\bm{Y}_{1:N},\bm{Y}_0\right)\\
&=&
\log\det\left(\bm{I}+\textrm{diag}\left(\bm{\Phi}_1^{-1},\cdots,\bm{\Phi}_N^{-1}\right)\bm{R}_{\bm{Y}_{1:N}|\bm{Y}_0}\right).
\nonumber
\end{eqnarray}
\section{Proof of Upper Bound 2}\label{appen:ub2}
To prove the statement, we first rewrite the objective and
constraint of (\ref{eq:rate_def}) as (\ref{proof:prop1_0}) and
(\ref{proof:prop1_2}), respectively. At the optimum point of
maximization (\ref{eq:rate_def}), the constraint is satisfied.
Therefore,
$I\left(\bm{Y}_{1:N};\hat{\bm{Y}}_{1:N}|\bm{Y}_0\right)\leq
\mathrm{R}$, which plugged into (\ref{proof:prop1_2}) obtains
\begin{eqnarray}\label{proof:ub2_2}
I\left(\bm{X}_s;\hat{\bm{Y}}_{1:N}|\bm{Y}_0\right) \leq \mathrm{R} -
I\left(\bm{Y}_{1:N};\hat{\bm{Y}}_{1:N}|\bm{Y}_0,\bm{X}_s\right),
\end{eqnarray} which in turn introduced into (\ref{proof:prop1_0}) allows
to bound
\begin{eqnarray}\label{proof:ub2_3}
I\left(\bm{X}_s;\bm{Y}_0,\hat{\bm{Y}}_{1:N}\right) \leq
I\left(\bm{X}_s;\bm{Y}_0\right)+\mathrm{R} -
I\left(\bm{Y}_{1:N};\hat{\bm{Y}}_{1:N}|\bm{Y}_0,\bm{X}_s\right)
\end{eqnarray} Since $I\left(\bm{Y}_{1:N};\hat{\bm{Y}}_{1:N}|\bm{Y}_0,\bm{X}_s\right)\geq
0$ by definition, we can state that $
I\left(\bm{X}_s;\bm{Y}_0,\hat{\bm{Y}}_{1:N}\right) \leq
I\left(\bm{X}_s;\bm{Y}_0\right)+\mathrm{R}$.

\section{Proof of Proposition \ref{Prop2}}\label{append:prop2}

In this Appendix, we solve the non-convex optimization
(\ref{eq:rate_singleBS_2}). Let us first expand:
\begin{eqnarray}\label{app:expad_logarithm}
&&\log\det\left(\bm{I}+\frac{\bm{Q}}{\sigma_r^2}\bm{H}_{s,0}^\dagger\bm{H}_{s,0}+\bm{Q}\bm{H}_{s,1}^\dagger\left(\bm{A}_1\sigma_r^2+\bm{I}\right)^{-1}\bm{A}_1\bm{H}_{s,1}
\right) \nonumber\\&&\ \ =
\log\det\left(\bm{I}+\frac{\bm{Q}}{\sigma_r^2}\bm{H}_{s,0}^\dagger\bm{H}_{s,0}\right)
+
\log\det\left(\bm{I}+\left(\bm{A}_1\sigma_r^2+\bm{I}\right)^{-1}\bm{A}_1\left(\bm{R}_{\bm{Y}_1|\bm{Y}_0}-\sigma_r^2\bm{I}\right)\right)\nonumber\\
&&\ \
=\log\det\left(\bm{I}+\frac{\bm{Q}}{\sigma_r^2}\bm{H}_{s,0}^\dagger\bm{H}_{s,0}\right)
+
\log\det\left(\bm{I}+\bm{A}_1\bm{R}_{\bm{Y}_1|\bm{Y}_0}\right)-\log\det\left(\bm{I}+\bm{A}_1\sigma_r^2\right).
\ \
\end{eqnarray}First equality follows from the value of $\bm{R}_{\bm{Y}_1|\bm{Y}_0}$ in (\ref{app:cond_cova_singleuser_YnYo_final}).
Notice that
$\log\det\left(\bm{I}+\frac{\bm{Q}}{\sigma_r^2}\bm{H}_{s,0}^\dagger\bm{H}_{s,0}\right)$
does not depend on $\bm{A}_1$. Therefore, the Lagrangian for the
problem can be written as
\begin{eqnarray}\label{lagrangian}
\mathcal{L}\left(\bm{A}_1,\lambda,\bm{\Phi}\right)=\left(1-\lambda\right)\log\det\left(\bm{I}+\bm{A}_1\bm{R}_{\bm{Y}_1|\bm{Y}_0}\right)-\log\det\left(\bm{I}+\bm{A}_1\sigma_r^2\right)
+ \lambda \mathrm{R} -
\mathrm{tr}\left\{\bm{\Phi}\bm{A}_1\right\}\nonumber ,
\end{eqnarray} where $\lambda$ is the Lagrange multiplier for the explicit constraint and $\bm{\Phi}\preceq0$ for the semidefinite positiveness constraint. The derivative of the Lagrangian with respect to $\bm{A}_1$ thus reads \cite{MATRIXCOOKBOOK}:
\begin{eqnarray}\label{eq:der_lagrangian}
&&\left[\frac{\partial\mathcal{L}}{\partial\bm{A}_1}\right]^T=\left(1-\lambda\right)\bm{R}_{\bm{Y}_1|\bm{Y}_0}\left(\bm{I}+\bm{A}_1\bm{R}_{\bm{Y}_1|\bm{Y}_0}\right)^{-1}-\sigma_r^2\left(\bm{I}+\bm{A}_1\sigma_r^2\right)^{-1}
-\bm{\Phi}.
\end{eqnarray}
Accordingly, the KKT conditions for the problem, which are
necessary but not sufficient, are:
\begin{eqnarray}\label{eq:prop2_kkts}
&&\textit{i})\ \
\left[\frac{\partial\mathcal{L}}{\partial\bm{A}_1}\right]^T=
\bm{0}\\
&&\textit{ii})\ \
\lambda\left(\log\det\left(\bm{I}+\bm{A}_1\bm{R}_{\bm{Y}_1|\bm{Y}_0}\right)-\mathrm{R}\right)=0
\nonumber\\
&&\textit{iii})\ \
\mathrm{tr}\left\{\bm{\Phi}\bm{A}_1\right\}=0.\nonumber
\end{eqnarray}Let now the eigen-decomposition $\bm{R}_{\bm{Y}_1|\bm{Y}_0}=\bm{U}\bm{S}\bm{U}^\dagger$. Then, it can be readily shown that matrix $\bm{A}_1^*=\bm{U}\textrm{diag}\left(\eta_1,\cdots,\eta_{N_1}\right)\bm{U}^\dagger$, with
\begin{eqnarray}\label{app:noiseallocation}
\eta_j=\left[\frac{1}{\lambda^*}\left(\frac{1}{\sigma_r^2}-\frac{1}{s_j}\right)-\frac{1}{\sigma_r^2}\right]^+,
\end{eqnarray} satisfies the KKT conditions, with multiplier $\lambda^*$ such that
$\sum_{j=1}^{N_1}\log\left(1+\eta_j s_j\right)=\mathrm{R}$
(therefore, $\lambda^*<1$), and multiplier $\bm{\Phi}^*\preceq 0$
computed from (\ref{eq:der_lagrangian}). Let now show that
$\bm{A}_1^*$ satisfies also the general sufficiency condition for
optimality, which is presented in the next Lemma.

\begin{lemma}\cite[Proposition 3.3.4]{BERTSEKAS} Let the differentiable maximization (\ref{eq:rate_singleBS_2}).
Consider a pair $\left(\bm{A}_1^*,\lambda^*\right)$ for which
$\lambda^*\left(\log\det\left(\bm{I}+\bm{A}_1^*\bm{R}_{\bm{Y}_1|\bm{Y}_0}\right)-
\mathrm{R}\right)=0$. Then, $\bm{A}_1^*$ is the global maximum of
(\ref{eq:rate_singleBS_2}) if:
\begin{eqnarray}\label{lagrangian_lemma} \bm{A}_1^*\in
\arg\max_{\bm{A}_1\succeq0}\mathcal{L}\left(\bm{A}_1,\lambda^*\right),
\end{eqnarray}where the Lagrangian\footnote{Notice that the semi-definite
multiplier $\bm{\Phi}$ has been removed of the Lagrangian by
constraining the maximization (\ref{lagrangian_lemma}) to the set
$\bm{A}_1\succeq0$.} has been defined in (\ref{lagrangian}).
\end{lemma}

\begin{lemma}\label{Lemma2}Let $\bm{A},\bm{B}\succeq0$, with ordered eigenvalues
$\bm{\Gamma}_A,\bm{\Gamma}_B$ respectively. Then,
\begin{eqnarray} \log\det\left(\bm{I}+\bm{A}\bm{B}\right)\leq
\log\det\left(\bm{I}+\bm{\Gamma}_A\bm{\Gamma}_B\right),
\end{eqnarray} with equality whenever $\bm{A}$ and $\bm{B}$ have
conjugate transpose eigenvectors.
\end{lemma}

\begin{proof} It is known that $\log\det\left(\bm{I}+\bm{A}\bm{B}\right)=\log\det\left(\bm{I}+\bm{\Gamma}_{AB}\right)$,
where $\bm{\Gamma}_{AB}$ are the ordered eigenvalues of
$\bm{A}\bm{B}$. Those eigenvalues are logarithmically majorized
\cite[Definition 1.4]{GUAN07} by the product of the separate
eigenvalues of $\bm{A}$ and $\bm{B}$, i.e.,
$\bm{\Gamma}_{AB}\prec_{\times}\bm{\Gamma}_{A}\bm{\Gamma}_{B}$
\cite[Theorem 9.H.1.d]{MARSHALL79}. Let now the function
$f\left(\bm{X}\right)=\log\det\left(\bm{I}+\bm{X}\right)$ be
defined on the set of semi-definite positive diagonal matrices,
\textit{i.e.,} $f\left(\bm{X}\right)=\sum\log\left(1+x_i\right)$.
We may apply \cite[Theorem 1.6]{GUAN07} to prove that
$f\left(\bm{X}\right)$ is a Schur-geometrically-convex function.
Accordingly, provided that
$\bm{\Gamma}_{AB}\prec_{\times}\bm{\Gamma}_{A}\bm{\Gamma}_{B}$,
then $
\log\det\left(\bm{I}+\bm{\Gamma}_{AB}\right)\leq\log\det\left(\bm{I}+\bm{\Gamma}_{A}\bm{\Gamma}_{B}\right)$,
which concludes the proof.\end{proof}

Let us prove now that our pair $\left(\bm{A}_1^*,\lambda^*\right)$
satisfies (\ref{lagrangian_lemma}). The lagrangian is defined for
the problem as
\begin{eqnarray}\label{lagra85}\mathcal{L}\left(\bm{A}_1,\lambda^*\right)=\left(1-\lambda^*\right)\log\det\left(\bm{I}+\bm{A}_1\bm{R}_{\bm{Y}_1|\bm{Y}_0}\right)-\log\det\left(\bm{I}+\bm{A}_1\sigma_r^2\right)
+ \lambda^* \mathrm{R}.\end{eqnarray} Recall that $\lambda^*<1$
and $\bm{R}_{\bm{Y}_1|\bm{Y}_0}=\bm{U}\bm{S}\bm{U}^\dagger$. Then,
using Lemma \ref{Lemma2} we can bound:
\begin{eqnarray}\label{app:indi_maxi}
\max_{\bm{A}_1\succeq0}\mathcal{L}\left(\bm{A}_1,\lambda^*\right)&\leq&\max_{\bm{\eta}\succeq0}\left(1-\lambda^*\right)\log\det\left(\bm{I}+\bm{\eta}\bm{S}\right)-\log\det\left(\bm{I}+\bm{\eta}\sigma_r^2\right)+ \lambda^* \mathrm{R}\nonumber\\
&=&  \lambda^* \mathrm{R}+\sum_{j=1}^{N_1}
\max_{\eta_j\geq0}\left(1-\lambda^*\right)\log\left(1+\eta_js_j\right)-\log\left(1+\eta_j\sigma_r^2\right)
\end{eqnarray}where $\bm{\eta}$ is the diagonal matrix of ordered eigenvalues of
$\bm{A}_1$. The individual maximizations on $\eta_j$ in
(\ref{app:indi_maxi}) are not concave. However, the continuously
differentiable functions
$f_j\left(\eta_j\right)=\left(1-\lambda^*\right)\log\left(1+\eta_js_j\right)-\log\left(1+\eta_j\sigma_r^2\right)$
have only two stationary points, i.e.,: \begin{eqnarray}
\frac{df_j}{d\eta_j}=0\rightarrow\left\{\begin{array}{c}
  \eta_j = \infty \\
  \eta_j = \frac{1}{\lambda^*}\left(\frac{1}{\sigma_r^2}-\frac{1}{s_j}\right)-\frac{1}{\sigma_r^2} \\
\end{array}\right.
\end{eqnarray} Recalling that $0\leq\lambda^*< 1$, it is easy to show that $\lim_{\eta_j\rightarrow \infty}f_j\left(\eta_j\right)=-\infty$. Therefore $\eta_j=\infty$ is the
global minimum of the problem. Considering the other stationary
point, it can be shown that its second derivative is lower than
zero. Accordingly, it is a local maximum, unique because there is
no other. However,  we restricted the optimization to the values
$\eta_j\geq0$. Hence, functions $f_j\left(\eta_j\right)$ take
maximum at:
\begin{eqnarray} \eta_j^* =\left[
\frac{1}{\lambda^*}\left(\frac{1}{\sigma_r^2}-\frac{1}{s_j}\right)-\frac{1}{\sigma_r^2}\right]^+.
\end{eqnarray}
Plugging this optimal values into (\ref{app:indi_maxi}), we bound
\begin{eqnarray}\label{app:indi_maxi2}
\max_{\bm{A}_1\succeq0}\mathcal{L}\left(\bm{A}_1,\lambda^*\right)\leq
\lambda^*
\mathrm{R}+\left(1-\lambda^*\right)\sum_{j=1}^{N_1}\log\left(1+\eta_j^*s_j\right)-\sum_{i=1}^N\log\left(1+\eta_j^*\sigma_r^2\right)
\end{eqnarray} Furthermore, noticing that for
$\bm{A}_1^*=\bm{U}\bm{\eta}^*\bm{U}^\dagger$:
\begin{eqnarray}\label{app:indi_maxi3}
\mathcal{L}\left(\bm{A}_1^*,\lambda^*\right)=\lambda^*
\mathrm{R}+\left(1-\lambda^*\right)\sum_{j=1}^{N_1}\log\left(1+\eta_j^*s_j\right)-\sum_{i=1}^N\log\left(1+\eta_j^*\sigma_r^2\right),
\end{eqnarray} then, it is demonstrated that
$\bm{A}_1^*=\arg\max_{\bm{A}_1\succeq0}\mathcal{L}\left(\bm{A}_1,\lambda^*\right)$.
Hence, the general sufficient condition holds, and it is optimum.
Finally, $\bm{\Phi}_1^*=\left(\bm{A}_1^{*}\right)^{-1}$, which
concludes the proof.

\section{}

\subsection{Proof of Proposition \ref{Prop3}}\label{append:prop3}

In this Appendix, we solve the  non-convex optimization
$\bm{A}_n^*=\arg\max_{\bm{A}_n\succeq0}
\mathcal{L}\left(\bm{A}_1,\cdots,\bm{A}_N,\lambda\right)$. First,
recall that
$\log\det\left(\bm{I}+\textrm{diag}\left(\bm{A}_1,\cdots,\bm{A}_N\right)\bm{R}_{\bm{Y}_{1:N}|\bm{Y}_0}\right)$
is equal to
${I\left(\bm{Y}_{1:N};\hat{\bm{Y}}_{1:N}|\bm{Y}_0\right)}$ (as
shown in (\ref{app:rate_prop2}), changing
$\bm{\Phi}_n=\bm{A}_n^{-1}$ $\forall \ n$). Then:
\begin{eqnarray}\label{constraint_extended}
\log\det\left(\bm{I}+\textrm{diag}\left(\bm{A}_1,\cdots,\bm{A}_N\right)\bm{R}_{\bm{Y}_{1:N}|\bm{Y}_0}\right)
&=& {I\left(\bm{Y}_{1:N};\hat{\bm{Y}}_{1:N}|\bm{Y}_0\right)}\\
&=&
{I\left(\bm{Y}_{1:N};\hat{\bm{Y}}_{n}^c|\bm{Y}_0\right)+I\left(\bm{Y}_{1:N};\hat{\bm{Y}}_{n}|\bm{Y}_0,\hat{\bm{Y}}_{n}^c\right)}\nonumber\\
&=&
{I\left(\bm{Y}_{n}^c;\hat{\bm{Y}}_{n}^c|\bm{Y}_0\right)+I\left(\bm{Y}_{n};\hat{\bm{Y}}_{n}|\bm{Y}_0,\hat{\bm{Y}}_{n}^c\right)}\nonumber\\
&=&
\log\det\left(\bm{I}+\textrm{diag}\left(\bm{A}_1,\cdots,\bm{A}_{n-1},\bm{A}_{n+1},\cdots,\bm{A}_N\right)\bm{R}_{\bm{Y}_{n}^c|\bm{Y}_0}\right)\nonumber\\
&& +
\log\det\left(\bm{I}+\bm{A}_n\bm{R}_{\bm{Y}_{n}|\bm{Y}_0,\hat{\bm{Y}}_{n}^c}\right)\nonumber
\end{eqnarray} where second equality follows from the chain rule for
mutual information, and the third from the Markov chain in
Proposition \ref{Def_berger_tung}. Finally, the fourth equality
evaluates the mutual information as in (\ref{app:rate_prop2}),
with $\bm{\Phi}_n=\bm{A}_n^{-1}$. The conditional covariances are
computed in Appendix \ref{appen:covariances}. Later, using
(\ref{app:cond_cova_singleuser_YnYo_ync_final}) and equivalently
to (\ref{app:expad_logarithm}):
\begin{eqnarray}\label{app:expand_loga_N}
&&\log\det\left(\bm{I}+\frac{\bm{Q}}{\sigma_r^2}\bm{H}_{s,0}^\dagger\bm{H}_{s,0}+\bm{Q}\sum_{n=1}^N\bm{H}_{s,n}^\dagger\left(\bm{A}_n\sigma_r^2+\bm{I}\right)^{-1}\bm{A}_n\bm{H}_{s,n}
\right)\nonumber\\
&&\qquad
=\log\det\left(\bm{I}+\frac{\bm{Q}}{\sigma_r^2}\bm{H}_{s,0}^\dagger\bm{H}_{s,0}+\bm{Q}\sum_{j\neq
n}\bm{H}_{s,j}^\dagger\left(\bm{A}_j\sigma_r^2+\bm{I}\right)^{-1}\bm{A}_j\bm{H}_{s,j}\right)\nonumber\\
&&\qquad\qquad +
\log\det\left(\bm{I}+\bm{A}_n\bm{R}_{{\bm{Y}}_{n}|\hat{\bm{Y}}_{n}^c,\bm{Y}_0}\right)-\log\det\left(\bm{I}+\bm{A}_n\sigma_r^2\right).
\end{eqnarray}

Therefore, plugging (\ref{constraint_extended}) and
(\ref{app:expand_loga_N}) into (\ref{eq:f_dual}), we can expand
the function under study as:
\begin{eqnarray}\label{app:lagrangian}
\mathcal{L}\left(\bm{A}_1,\cdots,\bm{A}_N,\lambda\right)
=\log\det\left(\bm{I}+\frac{\bm{Q}}{\sigma_r^2}\bm{H}_{s,0}^\dagger\bm{H}_{s,0}+\bm{Q}\sum_{j\neq
n}^N\bm{H}_{s,j}^\dagger\left(\bm{A}_j\sigma_r^2+\bm{I}\right)^{-1}\bm{A}_j\bm{H}_{s,j}\right)\qquad \qquad \\
+
\log\det\left(\bm{I}+\bm{A}_n\bm{R}_{{\bm{Y}}_{n}|\hat{\bm{Y}}_{n}^c,\bm{Y}_0}\right)-\log\det\left(\bm{I}+\bm{A}_n\sigma_r^2\right)\qquad
\qquad  \qquad
\nonumber\\
-
\lambda\left(\log\det\left(\bm{I}+\textrm{diag}\left(\bm{A}_1,\cdots,\bm{A}_{n-1},\bm{A}_{n+1},\cdots,\bm{A}_N\right)\bm{R}_{\bm{Y}_{n}^c|\bm{Y}_0}\right)
+
\log\det\left(\bm{I}+\bm{A}_n\bm{R}_{\bm{Y}_{n}|\hat{\bm{Y}}_{n}^c,\bm{Y}_0}\right)-\mathrm{R}\right)\nonumber
\end{eqnarray}
In order to obtain $\bm{A}_n^*=\arg\max_{\bm{A}_n\succeq0}
\mathcal{L}\left(\bm{A}_1,\cdots,\bm{A}_N,\lambda\right)$, we
first notice that the following Lagrangian
\begin{eqnarray}\label{app:eq:lagra}
\bar{\mathcal{L}}\left(\bm{A}_n,\lambda\right)=\left(1-\lambda\right)\log\det\left(\bm{I}+\bm{A}_n\bm{R}_{\bm{Y}_{n}|\bm{Y}_0,\hat{\bm{Y}}_{n}^c}\right)-\log\det\left(\bm{I}+\bm{A}_n\sigma_r^2\right)
+ \lambda \mathrm{R}
\end{eqnarray}satisfies
$\arg\max_{\bm{A}_n\succeq0}\bar{\mathcal{L}}\left(\bm{A}_n,\lambda\right)=\arg\max_{\bm{A}_n\succeq0}\mathcal{L}\left(\bm{A}_1,\cdots,\bm{A}_N,\lambda\right)$,
and it is identical to the Lagrangian in (\ref{lagra85}).
Therefore, we can directly apply derivation
(\ref{lagra85})-(\ref{app:indi_maxi3}) to solve it:

Consider first $\lambda \geq 1$. For it,
$\left(1-\lambda\right)\log\det\left(\bm{I}+\bm{A}_n\bm{R}_{\bm{Y}_{n}|\bm{Y}_0,\hat{\bm{Y}}_{n}^c}\right)-\log\det\left(\bm{I}+\bm{A}_n\sigma_r^2\right)\leq
0$, $\forall \bm{A}_n\succeq0$. Therefore, it is readily shown
that:
\begin{eqnarray}
\bm{0}=\arg\max_{\bm{A}_n\succeq0}
\mathcal{L}\left(\bm{A}_1,\cdots,\bm{A}_N,\lambda\right) \ \
\textrm{for} \ \lambda \geq 1.
\end{eqnarray} Let now $\lambda < 1$. Applying
(\ref{lagra85})-(\ref{app:indi_maxi3}) we show that
\begin{eqnarray}
\bm{U}_n\bm{\eta}\bm{U}_n^\dagger=\arg\max_{\bm{A}_n\succeq0}
\mathcal{L}\left(\bm{A}_1,\cdots,\bm{A}_N,\lambda\right) \ \
\textrm{for} \ \lambda < 1,
\end{eqnarray}with
$\bm{R}_{\bm{Y}_{n}|\bm{Y}_0,\hat{\bm{Y}}_{n}^c}=\bm{U}_n\bm{S}\bm{U}_n^\dagger$,
and
\begin{eqnarray}\label{app:noiseallocation}
\eta_j=\left[\frac{1}{\lambda}\left(\frac{1}{\sigma_r^2}-\frac{1}{s_j}\right)-\frac{1}{\sigma_r^2}\right]^+,\
\ j=1,\cdots,N_n.
\end{eqnarray}
This concludes the proof.


\subsection{Solution of (\ref{eq:g_dual}) with $\lambda \geq
1$}\label{landamayor1}

Applying equivalent arguments to those in
(\ref{app:expad_logarithm}), we can rewrite the Lagrangian in
(\ref{eq:g_dual}) as:
\begin{eqnarray}\label{lagrangian_lambda_mayor_1}
\mathcal{L}\left(\bm{A}_1,\cdots,\bm{A}_N,\lambda\right)&=&\left(1-\lambda\right)\log\det\left(\bm{I}+\textrm{diag}\left(\bm{A}_1,\cdots,\bm{A}_N\right)\bm{R}_{\bm{Y}_{1:N}|\bm{Y}_0}\right)\nonumber\\
&&-\log\det\left(\bm{I}+\textrm{diag}\left(\bm{A}_1,\cdots,\bm{A}_N\right)\sigma_r^2\right)-
\lambda \mathrm{R}\nonumber ,
\end{eqnarray} It is clear that, for $\lambda \geq 1$, the
Lagrangian takes its optimal value at
$\left\{\bm{A}_1^*,\cdots,\bm{A}_N^*\right\}=\bm{0}$.

\bibliographystyle{IEEEbib}
\bibliography{bibliography}


\begin{figure}[]
\center
\includegraphics[width=5.5in]{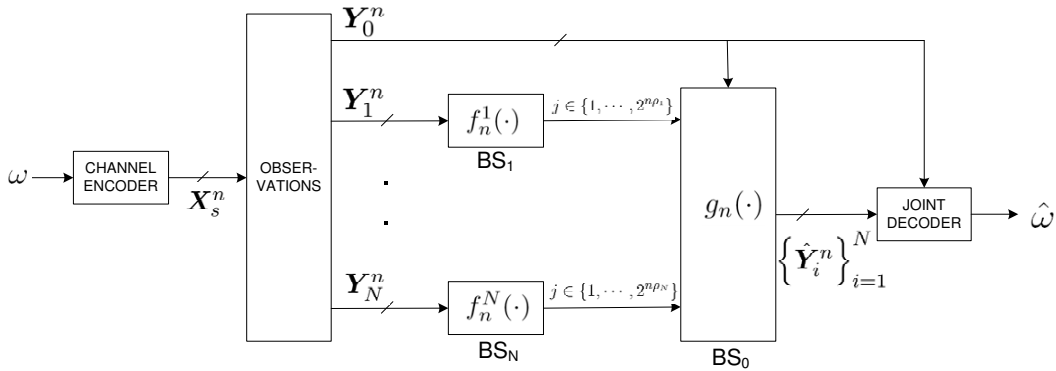}\vspace{-7mm}
\caption{Multiple-source compression with side information at the
decoder.} \label{fig:fig2}
\end{figure}



\begin{figure}[]
\vspace{-1mm} \center \subfigure[CDF versus R,
LOS]{\label{fig:fig4a}\includegraphics[width=3.3in]{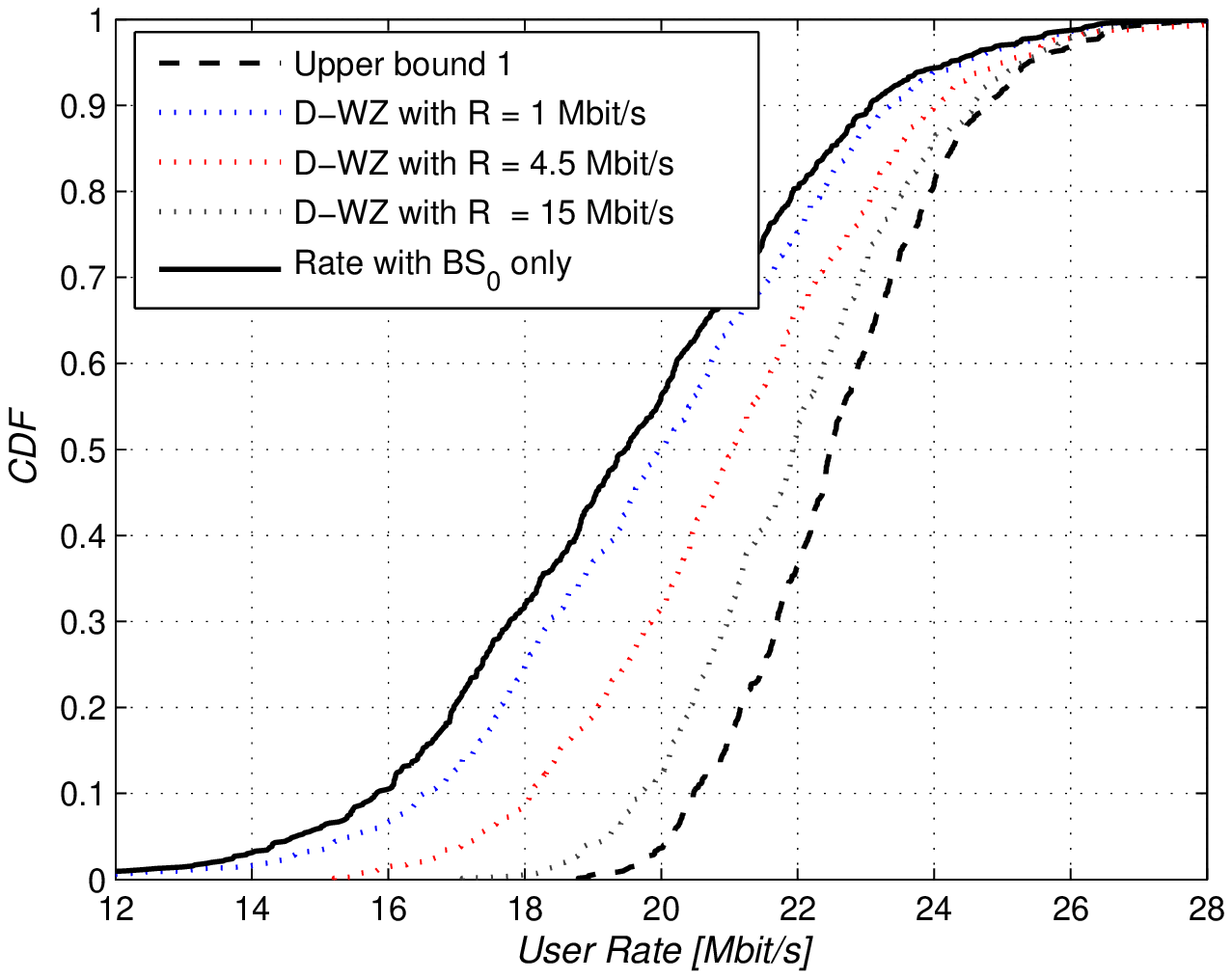}}
\hspace{-0.3in}\subfigure[CDF versus R,
N-LOS]{\label{fig:fig4b}\includegraphics[width=3.3in]{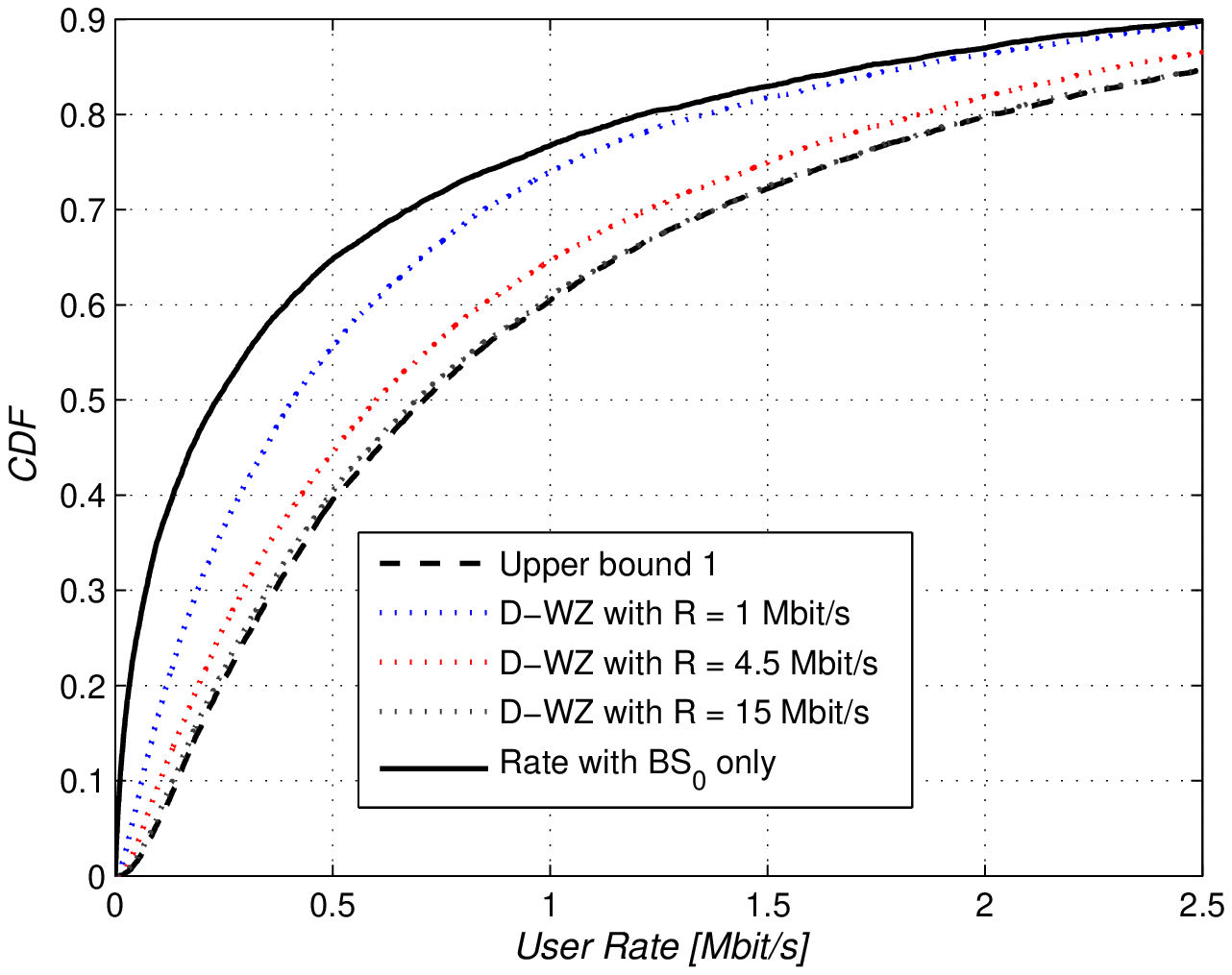}}\vspace{-2mm}
\caption{Single user capacity results with respect to the backhaul
rate. $\mathrm{BS}_1,\cdots,\mathrm{BS}_6$ cooperate with
$\mathrm{BS}_0$.}\label{fig:fig4}
\end{figure}
%
%
%

\begin{figure}[]
\vspace{-1mm} \center \subfigure[CDF versus N,
LOS]{\label{fig:fig6a}\includegraphics[width=3.3in]{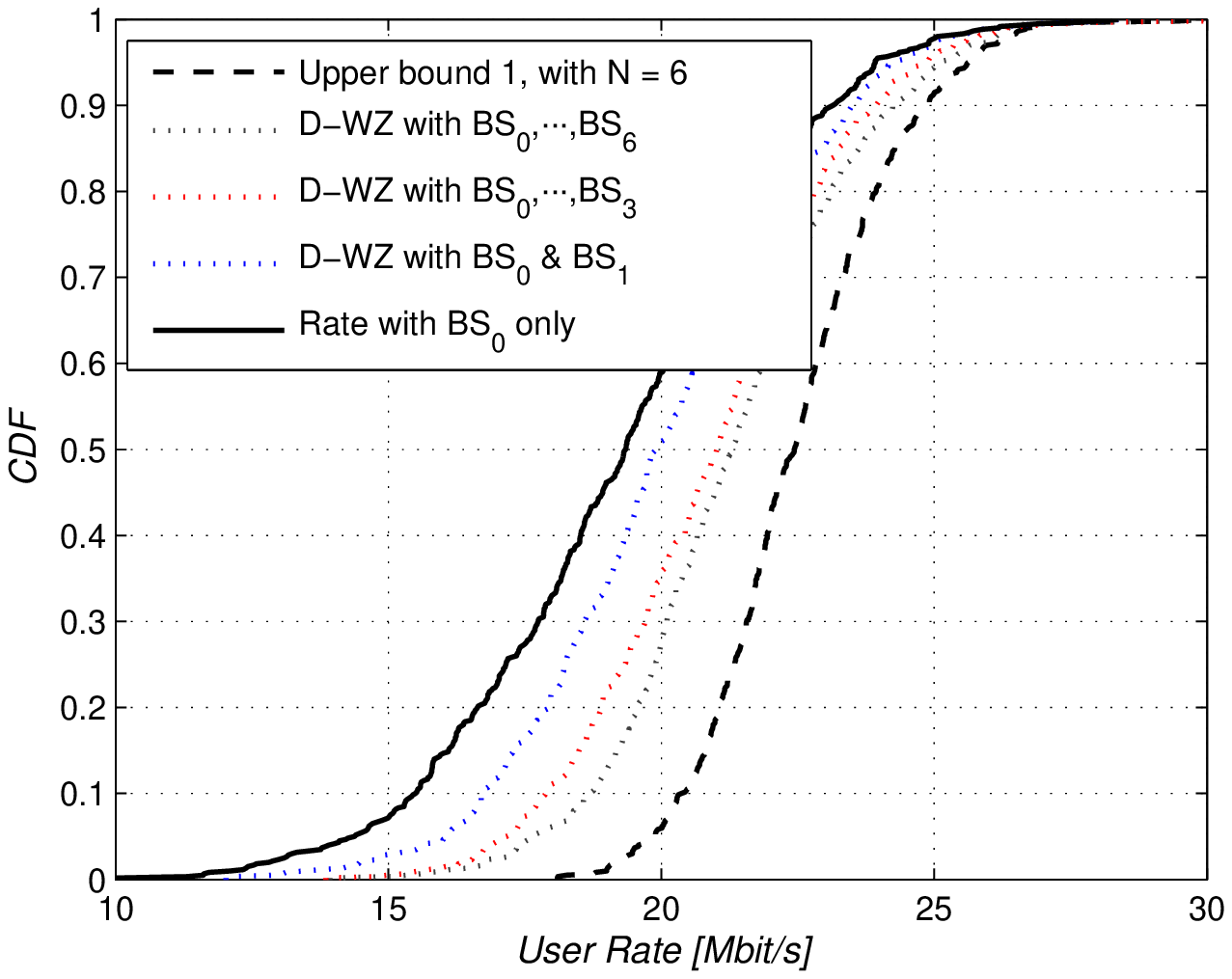}}
\hspace{-0.3in}\subfigure[CDF versus N,
N-LOS]{\label{fig:fig6b}\includegraphics[width=3.3in]{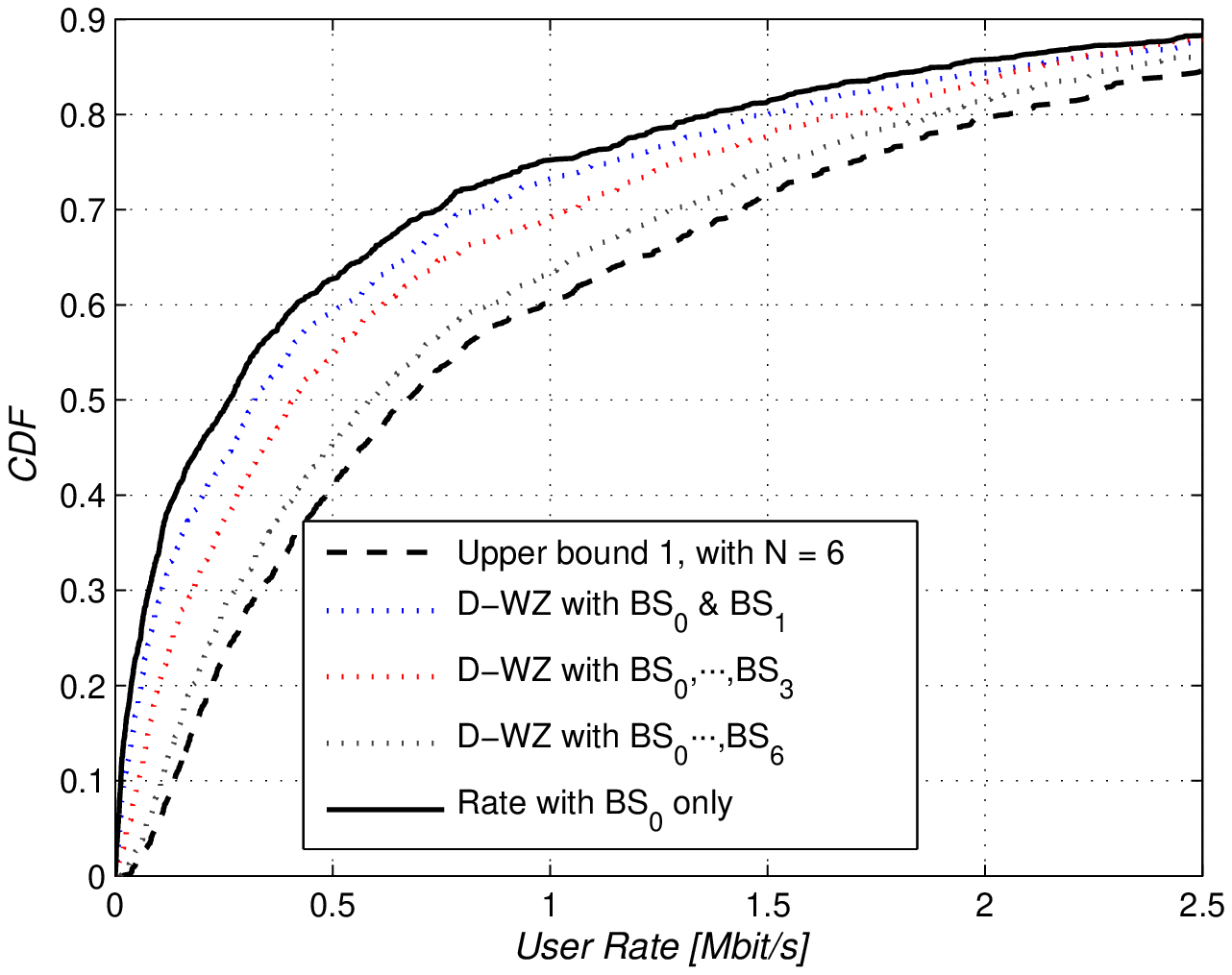}}\vspace{-2mm}
\caption{Single user capacity results with respect to the number
of Cooperative BS. Backhaul rate $\mathrm{R}=7$ Mbit/s
}\label{fig:fig6}
\end{figure}

\begin{figure}[]
\vspace{-10mm} \center
\includegraphics[width=3.3in]{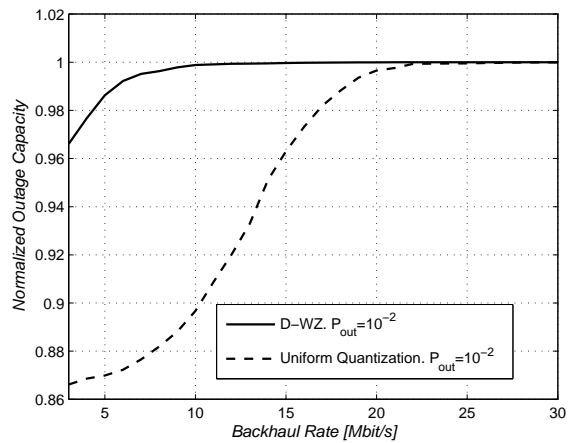} \vspace{-4mm}
\caption{Outage Capacity with D-WZ and with Quantization,
respectively, for different values of the backhaul rate
$\mathrm{R}$. LOS.} \label{fig:fig7}
\end{figure}


%
%
%
\begin{figure}[]
\vspace{-2mm} \center \subfigure[LOS
propagation]{\label{fig:fig8a}\includegraphics[width=3.3in]{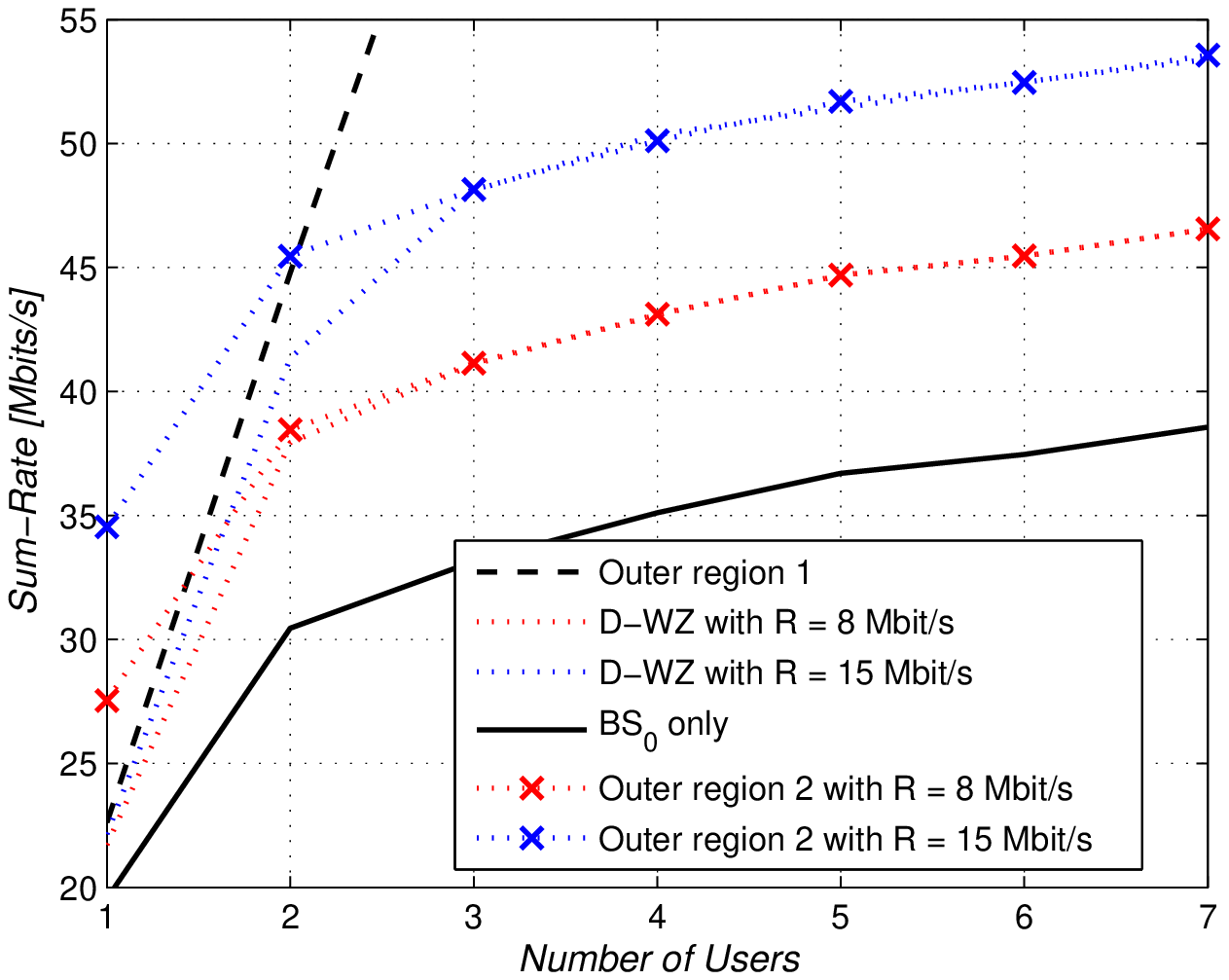}}
\hspace{-0.3in}\subfigure[N-LOS
propagation]{\label{fig:fig8b}\includegraphics[width=3.3in]{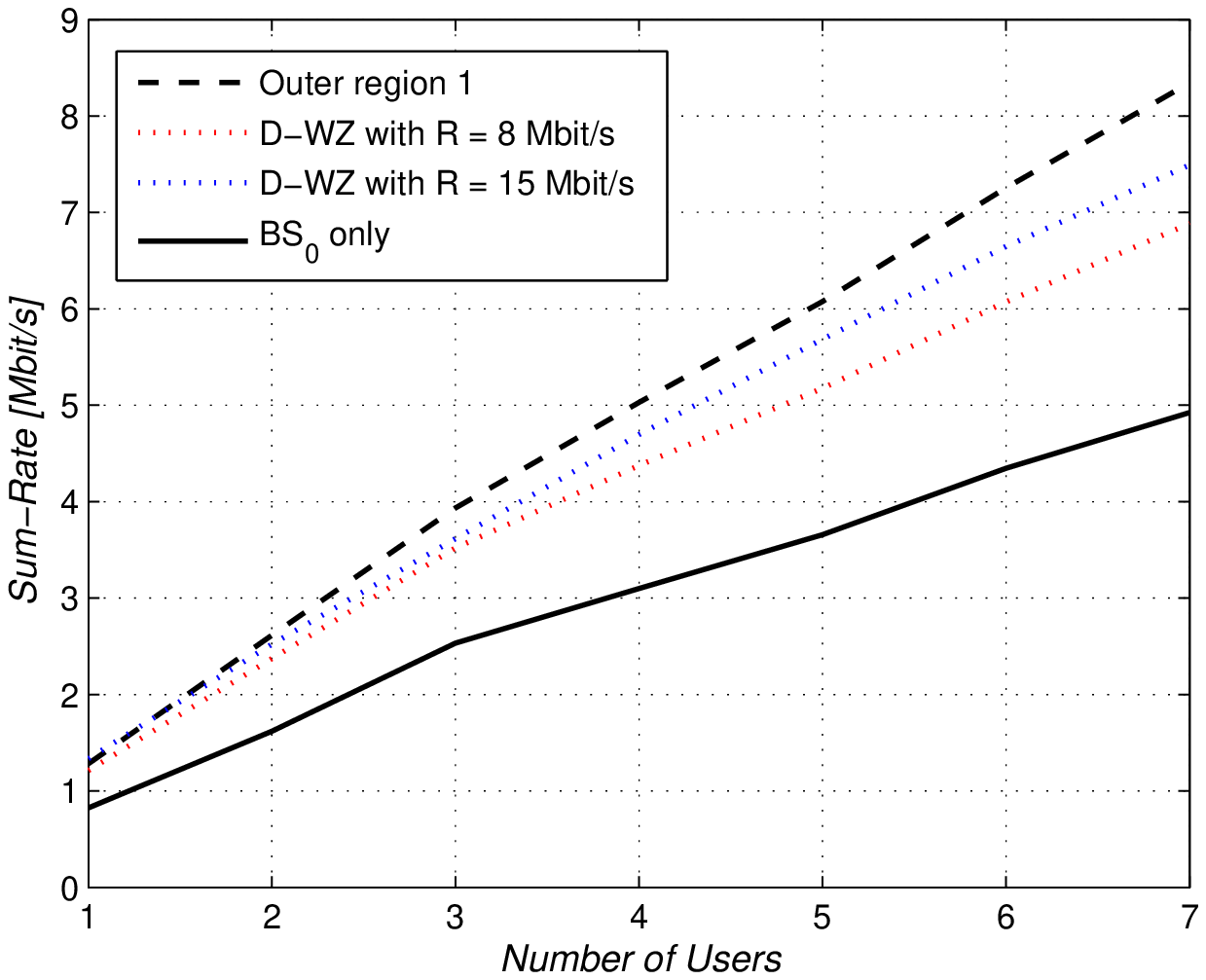}}\vspace{-3mm}
\caption{Sum-rate versus number of users.
$\mathrm{BS}_1,\cdots,\mathrm{BS}_6$ cooperate with
$\mathrm{BS}_0$. }\label{fig:fig8}
\end{figure}
%
%
%
\begin{figure}[]
\vspace{-2mm} \center \subfigure[LOS
propagation]{\label{fig:fig10}\includegraphics[width=3.3in]{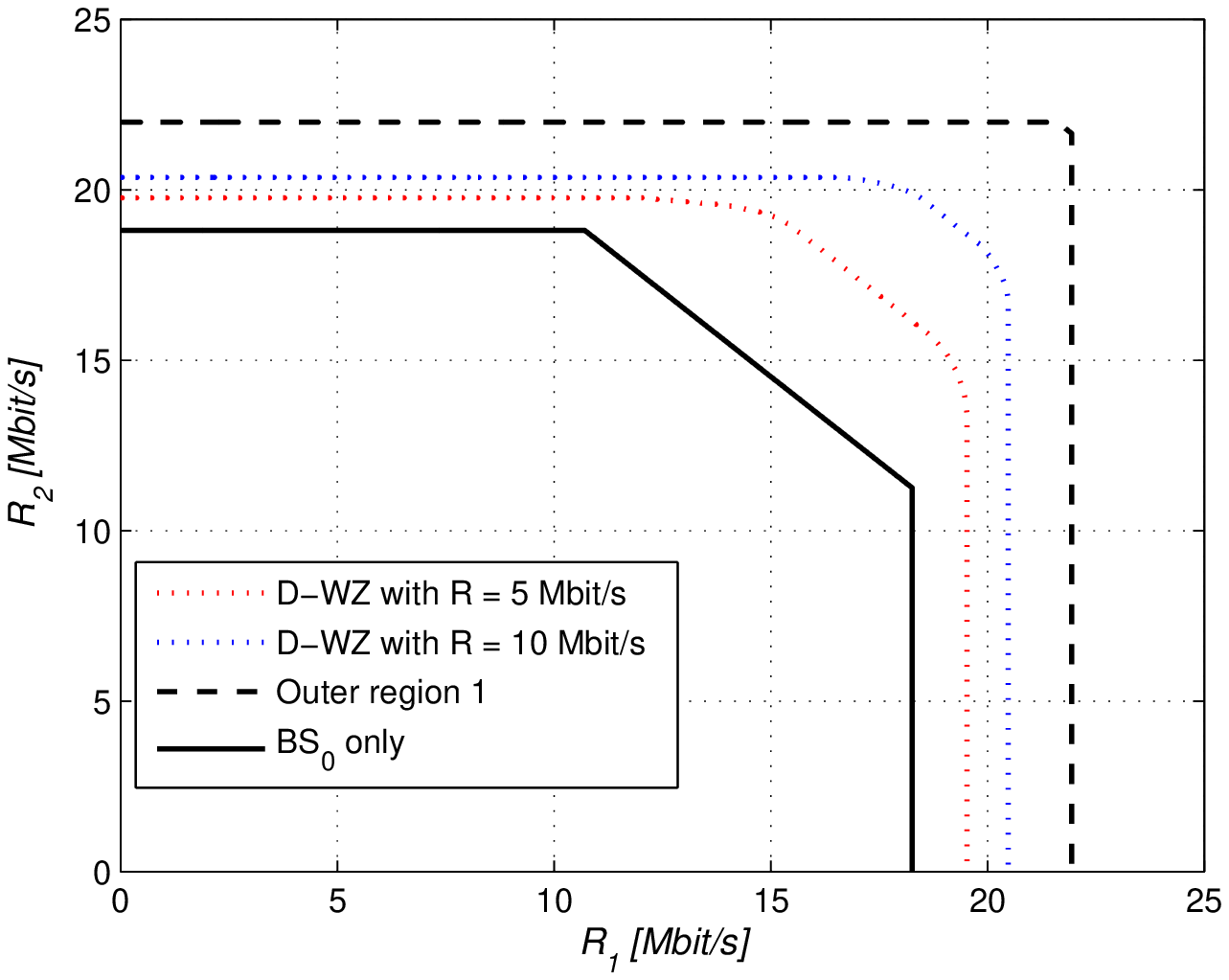}}
\hspace{-0.3in}\subfigure[N-LOS
propagation]{\label{fig:fig11}\includegraphics[width=3.3in]{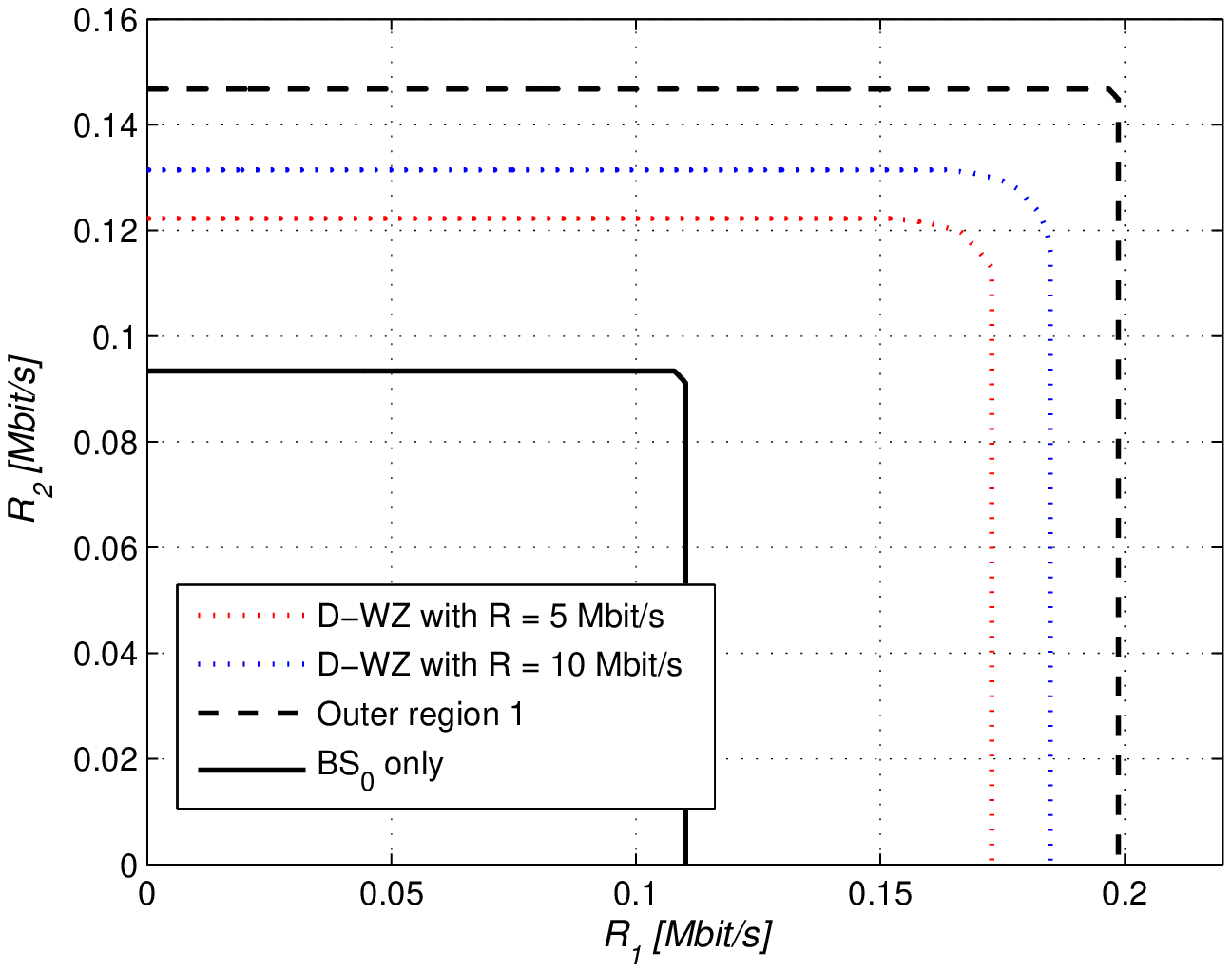}}\vspace{-3mm}
\caption{Rate region for different values of $\mathrm{R}$.
$\mathrm{BS}_1,\cdots,\mathrm{BS}_6$ cooperate with
$\mathrm{BS}_0$.}
\end{figure}
%
%

\end{document}